\documentclass[lettersize,journal]{IEEEtran}
\IEEEoverridecommandlockouts 

\usepackage{verbatim}
\usepackage{hyperref} 
\usepackage{graphicx, caption}
\usepackage{subcaption}
\usepackage{amssymb}
\usepackage{amsmath}
\usepackage{amsthm}
\usepackage{mathtools}
\usepackage{cite}
\usepackage{stfloats}
\usepackage{epstopdf}
\usepackage{psfrag}
\usepackage[mathscr]{euscript}
\usepackage{acronym}  
\usepackage{booktabs} 
\usepackage[linesnumbered,ruled,vlined]{algorithm2e}
\usepackage[table]{xcolor}
\usepackage{booktabs}

\SetKwInput{KwInput}{Input}
\SetKwInput{KwOutput}{Output}

\usepackage{color}
\usepackage{dsfont}
\usepackage{bbm}

\acrodef{IoT}{Internet of Things}
\acrodef{BS}{base station}
\acrodef{RSS}{received signal strength}
\acrodef{MSE}{mean square error}
\acrodef{MAE}{mean absolute error}
\acrodef{RMSE}{root mean square error}
\acrodef{LoS}{line-of-sight}
\acrodef{NLoS}{non-line-of-sight}
\acrodef{CNN}{convolutional neural network}
\acrodef{MLP}{multilayer perceptron}
\acrodef{ML}{machine learning}
\acrodef{FLOPs}{floating point operations}
\acrodef{mmWave}{millimeter wave}
\usepackage{color}
\usepackage{dsfont}
\usepackage{bbm}

\newtheorem{theorem}{Theorem}
\newtheorem{lemma}{Lemma}

\newcommand{\E}{\mathbb{E}}
\renewcommand{\P}{\mathbb{P}}
\newcommand{\ka}{\nonumber \\}

\newcommand{\domcom}{d_c}
\newcommand{\domspe}{d_{s_j}}
\newcommand{\domspelist}{D_s}
\newcommand{\dimfeature}{m}
\newcommand{\head}{\boldsymbol{v}}
\newcommand{\domcomweight}{w_c}
\newcommand{\domspeweight}{W_s}
\newcommand{\feature}{\boldsymbol{f}}
\newcommand{\medoutputnum}{N}
\newcommand{\mean}{\mu}

\newcommand{\covariance}{\Sigma}
\newcommand{\medmultimodalin}{\boldsymbol{X}}
\newcommand{\medout}{\boldsymbol{Y}}
\newcommand{\numagents}{K}
\newcommand{\nn}{\theta}
\newcommand{\medloss}{\mathcal{L}}
\newcommand{\phycoeff}{\lambda}

\newcommand{\distance}{W}
\newcommand{\antennanum}{U}
\newcommand{\vehiclenum}{V}
\newcommand{\txpower}{P_\text{tx}}
\newcommand{\txgain}{G_\text{tx}}
\newcommand{\rxgain}{G_\text{rx}}
\newcommand{\pathloss}{PL}
\newcommand{\carrier}{f_c}
\newcommand{\shadow}{S_f}
\newcommand{\reflection}{\bar{R}}
\newcommand{\blockage}{B}
\newcommand{\rxheight}{r_h}
\newcommand{\hypothesis}{\Theta}
\newcommand{\medbatchsize}{C}

\newcommand{\medepochs}{E}
\newcommand{\medlr}{\eta}
\newcommand{\trainnum}{m^\text{tr}}
\newcommand{\testnum}{m^\text{te}}
\newcommand{\rss}{R}

\allowdisplaybreaks

\title{Efficient Domain Generalization in Wireless Networks with Scarce Multi-Modal Data}
\author{  		\IEEEauthorblockN{
			Minsu~Kim, Walid~Saad \textit{Fellow, IEEE}, and Doru~Calin
			\vspace{-0.7cm}
		}
        
        		\thanks{
        			M.\ Kim and W.\ Saad are with the Institute for Advanced Computing,  Virginia Tech, Alexandria, VA, USA (email: \{\texttt{msukim, walids}\}@vt.edu.)
        			
        			D.\ Calin is with MediaTek, Warren, NJ, USA (eamil: \{\texttt{doru.calin@mediatek.com}\}).
        			
        			The implementation code will be available on https://github.com/news-vt.
        			
					A preliminary version of this work was submitted to IEEE International Conference on Communications (ICC) 2026.
		
				}
        
        }

\begin{document}
\maketitle

\begin{abstract}
In 6G wireless networks, multi-modal machine learning (ML) models can be leveraged to enable situation-aware network decisions in dynamic environments. However, trained ML models often fail to generalize under domain shifts when training and test data distributions are different because they often focus on modality-specific spurious features. In practical wireless systems, domain shifts occur frequently due to dynamic channel statistics, moving obstacles, or hardware configuration. To mitigate domain shifts, one can do extensive measurement or generate synthetic multi-modal data from various wireless environments. However, without access to target domains, this approach is highly time-consuming. Moreover, public multi-modal wireless datasets are extremely scarce unlike vision or language tasks. Thus, there is a need for learning frameworks that can achieve robust generalization under scarce multi-modal data in wireless networks. In this paper, a novel and data-efficient two-phase learning framework is proposed to improve generalization performance in unseen and unfamiliar wireless environments with minimal amount of multi-modal data. In the first stage, a physics-based loss function is employed to enable each base station (BS) to learn the physics underlying its wireless environment captured by multi-modal data. The data-efficiency of the physics-based loss function is analytically investigated. In the second stage, collaborative domain adaptation is proposed to leverage the wireless environment knowledge of multiple BSs to guide under-performing BSs under domain shift. Specifically, domain-similarity-aware model aggregation is proposed to utilize the knowledge of BSs that experienced similar domains. To validate the proposed framework, a new dataset generation framework is developed by integrating CARLA and MATLAB-based mmWave channel modeling to predict mmWave received signal strength from LiDAR, RGB, radar, and GPS. Simulation results show that the proposed physics-based training requires only 13\% of data samples to achieve the same performance as a state-of-the-art baseline that does not use physics-based training. Moreover, the proposed collaborative domain adaptation needs only 25\% of data samples and 20\% of FLOPs to achieve the convergence compared to baselines. 
\end{abstract}

\section{Introduction}
Next-generation wireless systems such as 6G can potentially leverage a broad range of sensing modalities, including LiDAR, RGB, radar, or GPS, so as to provide situational awareness \cite{10929033}. For instance, in vehicular networks, multi-modal sensing data can be exploited for various communication functions, such as beam forming and blockage prediction \cite{park2025resource}. Multi-modal sensing data can also be used for learning dynamic wireless environments such as path loss and \ac{RSS} estimation \cite{10949588}. To process multi-modal data, \ac{ML} models are widely adopted because neural networks can extract and fuse features from different modalities. Hence, fusing heterogeneous information sources can provide a wireless network with situational-awareness, such as moving blockages, obstacles, or trajectories that single-modality \ac{ML} schemes cannot fully provide under wireless environments.

Despite the potential benefits of using multi-modal data, conventional data-driven training methods (e.g, supervised or semi-supervised) require a significant amount of data. However, multi-modal wireless communications scenarios have typically scarce datasets compared to other \ac{ML} tasks such as vision and language. For instance, DeepSense 6G \cite{alkhateeb2023deepsense} has around one million data samples while LAION-5B \cite{schuhmann2022laion} has around five billion data samples. Hence, models often become overfitted to modality-specific spurious features such as noise, camera brightness, or LiDAR intensity patterns \cite{qu2023modality} due to training data scarcity. Hence, trained multi-modal models often fail to generalize to unseen environments when training and test data distributions are different \cite{zhang2025out}. Such a data distribution change after training is called \emph{domain shift}. In practical wireless environments, domain shifts can happen frequently due to inherently dynamic channel statistics, blockages, weather, and device-related characteristics. For example, consider a \ac{mmWave} beamforming task that should be performed based on multi-modal data input in vehicular scenarios as done in \cite{park2025resource}. Environmental change can affect each modality sensor. For example, a target vehicle with unfamiliar shape can appear while rain or snow can reduce the intensity of LiDAR and attenuate radar signals. Such domain shifts induce a change to the input distribution of the \ac{ML} model, and they are defined as \emph{covariate shift}. Meanwhile, large vehicles can block the \ac{mmWave} signal of the target vehicle, removing \ac{LoS} paths. A target vehicle can have different \ac{mmWave} antenna position on the roof or the windshield, changing the optimal beam indices completely. Such domain shifts are known as \emph{concept shift}, and they change the mapping between the input and output (i.e., beam index). Both covariate and concept shifts can significantly affect the performance of \ac{ML} models trained by conventional data-driven methods because learned mappings do not align with the new distributions. One can collect a sufficient amount of data from every possible wireless environment to mitigate domain shifts, but this approach is not practical. Without access to target domains, robust generalization requires extensive measurement or synthetic data generation, which is costly.  Therefore, there is a need to design new learning frameworks that can achieve robust generalization with minimal multi-modal data in wireless networks.

\subsection{Related Works}
There are a handful of works \cite{jiao2025addressing, jiao2025ai2mmum, chen2025analogical, liu2023wisr, 10207706, 10742098} that considered the problem of \ac{ML} generalization for wireless communication scenarios. In \cite{jiao2025addressing}, the authors proposed a two-stage multi-modal pre-training and downstream task adaptation framework to utilize one universal model for diverse downstream tasks. The work in \cite{jiao2025ai2mmum} used a large language model (LLM) to integrate language and radio frequency modalities to fine-tune it for diverse tasks. In \cite{chen2025analogical}, the authors proposed an analogical learning framework to allow a model to learn the scenario-related information that can be used across different scenarios. However, the works in \cite{jiao2025addressing, jiao2025ai2mmum, chen2025analogical} considered the generalization of a pre-trained model to different tasks rather than generalization to domain shifts caused by dynamic wireless environments. The authors in \cite{liu2023wisr} and \cite{10207706} considered generalization for gesture classification using wireless sensing under different subcarriers and scattering environments. In \cite{10742098}, the authors proposed a model for radio frequency fingerprinting to identify transmitters to generalize over various receivers. However, the works in \cite{liu2023wisr, 10207706, 10742098} only considered covariate shifts, while concept shifts are not studied. In practice, both domain shifts can occur frequently in wireless environments, and their generalization performance can be limited only under covariate shifts. Moreover, the works in \cite{liu2023wisr, 10207706, 10742098} assumed access to multiple training data domains, which does not always hold in practice due to data scarcity.

Meanwhile, the works in \cite{seretis2022toward, li2025digital, zheng2023cell, li2024map} used wireless models such as channel propagation as a prior physical information to improve the data efficiency. In \cite{seretis2022toward}, the authors used the geometry of indoor maps to predict \ac{RSS} by diving receivers into \ac{LoS} and \ac{NLoS}. The authors in \cite{li2025digital} proposed physics-informed reinforcement learning for \ac{mmWave} indoor navigation tasks with digital twins. The work in \cite{zheng2023cell} proposed an imaged-based wireless propagation based model based on generative adversarial networks to predict the \ac{RSS} of cells. In \cite{li2024map}, the authors used a 3GPP path loss formula for indoor \ac{RSS} prediction for \ac{LoS} receivers while used regression based methods for \ac{NLoS}. However, the solutions in \cite{seretis2022toward, li2025digital, zheng2023cell, li2024map} did not consider the robustness of their trained models to domain shifts caused by wireless environments. Moreover, the works in \cite{seretis2022toward, li2025digital, zheng2023cell, li2024map} also did not analyze the impact of using prior physical information during training on the data-efficiency. To the best of our knowledge, there are no current works that jointly consider multiple modalities, domain generalization, and data efficiency to mitigate domain shifts in wireless communication scenarios. 

\subsection{Contributions}
The main contribution of this paper is a novel data-efficient \ac{ML} framework that can promptly achieve robust domain generalization to unseen and unfamiliar wireless environments with minimal multi-modal data. The proposed framework consists of two phases: In the first phase, we train each \ac{BS}'s model by exploiting the channel propagation physics of the wireless environment captured by multi-modal input. Specifically, we propose to use a physics-based loss function to mitigate domain shifts and to improve data-efficiency. We theoretically prove how the physics-based loss function can improve data efficiency by regularizing the hypothesis space. In the second phase, we propose collaborative domain adaptation, which utilize multiple \ac{BS}s' knowledge of their wireless environment to help a \ac{BS} under domain shifts. Specifically, we first propose a method to measure the domain similarity between two \ac{BS}. Then, we propose domain-similarity-aware model aggregation to use the knowledge of other \ac{BS}s who experienced similar domains to adapt to domain shifts with minimal amount of data. To validate the proposed framework, we develop a new dataset generation framework, which predicts \ac{mmWave} \ac{RSS} using multi-modal sensor (LiDAR, RGB, Radar, and GPS) inputs in urban vehicular scenarios. We integrate CARLA \cite{dosovitskiy2017carla} and MATLAB for the data generation. Specifically, we use CARLA's autonomous vehicle functions to induce dynamic blockages and multi-modal sensing modules. Based on the generated CARLA environment, we use MATLAB to perform ray-tracing to calculate the \ac{RSS} of the coverage area of \acp{BS}. Simulation results show that our proposed framework can achieve significantly more robust generalization and data efficiency compared to the standard fine-tuning and state-of-the-art baselines that used physics-prior information. Specifically, our framework requires only 13\% of the training data to achieve the same performance as the baseline that does not use physics-based training. For the proposed collaborative domain adaptation, our method needs only 25\% of data samples and 20\% of \ac{FLOPs} to converge compared to baselines while the standard fine-tuning fails to converge due to lack of data. 

The rest of this paper is organized as follows. Section \ref{sec: med_system_model} presents the system model. In Section \ref{sec: med_use_case}, we present a use case to validate the proposed idea. Section \ref{sec: med_data_generation} provides the data generation framework. In Section \ref{sec: med_experiments}, simulation results are provided. Finally, conclusions are drawn in Section \ref{sec: med_conclusion}.

\section{System Model} \label{sec: med_system_model}
We consider $\numagents$ \ac{BS}s equipped with a uniform planar array (UPA) of $\antennanum$ antennas and a set $\mathcal{\vehiclenum}$ of $\vehiclenum$ vehicles distributed over a given geographical area. We assume that each \ac{BS} is equipped with multi-modal sensors including camera, LiDAR, radar, and GPS. Based on a captured multi-modal input $\medmultimodalin_k$, \ac{BS} $k$ it makes a network decision $\medout_k$. Without loss of generality, output $\medout$ can be a beamforming index, channel estimation, handover decision, or \ac{RSS} prediction of the surrounding area \cite{alkhateeb2023deepsense, jiao2025addressing, jiao2025ai2mmum, park2025resource}. To process multi-modal input $\medmultimodalin_k$, each \ac{BS} $k$ uses an \ac{ML} model $\nn_k$. In our system, a \emph{domain} is essentially a joint probability distribution $\P(\medmultimodalin, \medout)$ between input $\medmultimodalin$ and output $\medout$. Particularly, we define the joint probability distributions $\P^{\text{train}}_{k}(\medmultimodalin, \medout)$ and  $\P^{\text{test}}_{k}(\medmultimodalin, \medout)$ as the domain of training and test data, respectively. 
We define domain shift as any event that causes $\P^{\text{test}}_{k}(\medmultimodalin, \medout) \neq \P^{\text{train}}_{k}(\medmultimodalin, \medout)$. Specifically, a \emph{covariate shit} changes input distribution as $\P^{\text{train}}_{k}(\medmultimodalin) \neq \P^{\text{test}}_{k}(\medmultimodalin)$, and a \emph{concept shift} changes the mapping between $\medmultimodalin_k$ and $\medout_k$ as $\P^{\text{train}}_{k}(\medout | \medmultimodalin) \neq \P^{\text{test}}_{k}(\medout | \medmultimodalin)$.  In a wireless system, both covariate and concept shifts occur frequently. For example, consider a \ac{mmWave} beamforming task based on multimodal input $\medmultimodalin$. Covariate shifts can occur if a target vehicle has a new shape, affecting the vision domain, or under weather changes. Concept shifts can, for example, occur when a \ac{mmWave} antenna is mounted in a different position on the target vehicle, changing the optimal beam index. For a \ac{mmWave} \ac{RSS} prediction task, unseen blockages change channel propagation paths significantly, making a new mapping between $\medmultimodalin$ and $\medout$. Under such domain shifts, trained models often show severe performance degradation. 

Although collecting additional training data can mitigate domain shifts, such an approach is prohibitively time-consuming and costly. Moreover, under domain shifts, a trained model must also be adapted within a few epochs using significantly limited of data from target domains. Hence, there is a need to design a data-efficient multi-modal training and adaptation frameworks to achieve robust generalization to unseen wireless environments.

\subsection{Problem Formulation}
Given our model, we now formulate two multi-objective optimization problems for training and adaptation, respectively. For training, we aim to minimize the amount of training data $\trainnum_k$ and $\nn_k$ that can optimize the generalization error under unseen domain $\P_k^\text{test}(\medmultimodalin, \medout)$ for \ac{BS} $k, \forall k$ as follows:
\begin{subequations}
\begin{align}
 &\min_{\trainnum_k} \quad \left(
\trainnum_k, 
 \E_{(\medmultimodalin, \medout) \sim \P^\text{test}_k(\medmultimodalin, \medout)} [ l(\medout, \nn(\medmultimodalin)) ]
 \right) \label{prob: training}
  \\
 & \ \text{s.t.} \ \quad \nn = \arg\min_\nn \frac{1}{\trainnum_k}\sum_{i=1}^{\trainnum_k} l(\medout_i, \theta(\medmultimodalin_i)) \\
 & \quad \quad \ \ (\medout_i, \medmultimodalin_i )\sim \P^\text{train}_k(\medmultimodalin, \medout), \forall i,
\end{align}
\end{subequations}
where $l(\cdot)$ is a loss function. Here, $\P_k^\text{test}(\medmultimodalin, \medout)$ can be any unseen domain, such as different antenna positions on target vehicles, unfamiliar blockage patterns, or weather changes. This problem is challenging because we need to jointly minimize $\trainnum_k$ and the generalization error $l(\medout, \nn(\medmultimodalin))$ without access to $\P^\text{test}_k(\medmultimodalin, \medout)$. After training, we need to adapt the trained model when domain shifts occur to recover the performance. Hence, for the adaption, we aim to find optimal adapted model $\tilde{\nn}_k$ and minimal amount of data $\testnum_k$ from $\P^\text{test}_k(\medmultimodalin, \medout)$ to minimize the generalization error as follows:
\begin{subequations}
\begin{align}
	&\min_{\testnum_k, \tilde{\nn}_k^{(\medepochs)}} \quad \left(
	\testnum_k,  \E_{(\medmultimodalin, \medout) \sim \P^\text{test}_k(\medmultimodalin, \medout)} [ l(\medout, \tilde{\nn}_k^{(\medepochs)}(\medmultimodalin)) ]
	\right) \label{prob: adaptation}
	 \\
	&\ \text{s.t.} \quad \quad \tilde{\nn}_k^{(0)} \in \Theta'_k =\{\nn_k : \nn_k \text{\ that satisfies (1b)}\} \\
	&\ \quad \quad \quad \tilde{\nn}_k^{(e+1)} \hspace{-0.5mm} =  \hspace{-0.5mm} \tilde{\nn}_k^{(e)} - \medlr \nabla_{\tilde{\nn}_k} 
	\Big(
	\frac{1}{\testnum_k}\sum_{i=1}^{\testnum_k} l(\medout_i, \tilde{\nn}_k(\medmultimodalin_i))
	\Big)
	 \\
	& \ \quad \quad \quad (\medout_i, \medmultimodalin_i )\sim \P^\text{test}_k(\medmultimodalin, \medout), \forall i; \ e \in [0, \medepochs-1],
\end{align}
\end{subequations}
where $\medepochs$ is the number of epochs for adaptation. This problem is challenging since the amount of available data from $\P^\text{test}_k(\medmultimodalin, \medout)$ is significantly smaller than $\P^\text{train}_k(\medmultimodalin, \medout)$. Next, we present how to solve problems \eqref{prob: training} and \eqref{prob: adaptation} with the proposed physics-based training and collaborative domain adaptation.

\section{Proposed Physics-Based Training \\ and Collaborative Adaptation } \label{sec: proposed_framework}
To solve problems \eqref{prob: training} and \eqref{prob: adaptation}, we design a two-phase learning framework. For the first phase, we make predictions of a model consistent with channel propagation so as to embed domain-invariant physical knowledge into the training for \eqref{prob: training}. For the second phase, we use the knowledge of other trained \ac{BS}s to help under-performing \ac{BS}s under domain shifts based on domain similarity for \eqref{prob: adaptation}. 

\subsection{Physics-based Training} \label{subsec: physics_based_training}
In the first phase, we solve problem \eqref{prob: training} to minimize the generalization error with minimal training data. Essentially, without access to $\P^{\text{test}}_k(\medmultimodalin, \medout)$, model $\nn_k$ usually overfits to $\P^{\text{train}}_k(\medmultimodalin, \medout)$, thereby being susceptible to covariate and concept shifts. Model $\nn_k$ usually memorizes noise in $\P^{\text{train}}_k(\medmultimodalin, \medout)$ rather than learning the relationship between $\medmultimodalin$ and $\medout$. To mitigate this, we make predictions of model $\nn_k$ consistent with channel propagation in the wireless environment of \ac{BS} $k$. As such, model $\nn_k$ can be robust to covariate and concept shifts because its predictions are based on physical laws invariant over domains. To this end, we use a physics-based loss function to embed channel propagation characteristics into the training, as follows: 
\begin{align}
    &\medloss_{\text{total}}  \hspace{-0.3mm} (\medout_k, \hspace{-0.5mm} \nn_k \hspace{-0.3mm} (\medmultimodalin_k) \hspace{-0.3mm} ) \hspace{-0.75mm} = \hspace{-0.75mm} \medloss_\text{data} (\medout_k, \nn_k(\medmultimodalin_k)) \hspace{-0.75mm} + \hspace{-0.75mm} \phycoeff \medloss_\text{phy}(\medout_k, \nn_k(\medmultimodalin_k)), \label{physics_loss}
\end{align}
where $(\medmultimodalin_k, \medout_k) \sim \P^{\text{train}}_{k}(\medmultimodalin, \medout)$, $\medout_k$ is a ground truth output (e.g., beamforming index or \ac{RSS}), $\medloss_\text{data} (\medout_k^\text{true}, \nn_k(\medmultimodalin_k)) = \frac{1}{\medbatchsize} \sum_{i=1}^{\medbatchsize} l_{\text{data}}(\medout_{k, i}, \nn_k(\medmultimodalin_{k, i}))$ is an empirical loss function with batch size $\medbatchsize$, and $l_{\text{data}}$ is a supervised loss function such as \ac{MSE} or \ac{MAE}. Similarly, $\medloss_\text{phy}(\medout_k, \nn_k(\medmultimodalin_k))  = \frac{1}{\medbatchsize} \sum_{i=1}^{\medbatchsize} l_{\text{phy}}(\medout_{k,i}, \nn_k(\medmultimodalin_{k, i}))$, $l_{\text{phy}}(\medout_{k, i}, \nn_k(\medmultimodalin_{k, i}))$ is a physics-based loss function, and  $\phycoeff$ is the corresponding coefficient. Here, $\medloss_\text{phy}$ penalizes the model if its predictions do not follow the physical law of the channel propagation. Hence, $\medloss_\text{phy}$ can reduce the possible search space of $\nn_k$, improving data efficiency. To formalize this intuition, we first define a hypothesis class $\hypothesis_k$ of $\nn_k$. Without loss of generality, we assume $l_\text{data}$ and $l_\text{phy} \in [0, 1]$ as done in \cite{garg2020functional, kim2024spafl}\footnote{In practice, we can clip gradients to limit the impact of the loss function.}. To simplify the notation, we drop index $k$ hereinafter. For given $\tau \in [0,1]$ and a certain domain $\P(\medmultimodalin, \medout)$, we define $\hypothesis (\tau)$ as the subset of hypothesis $\hypothesis$ as follows
\begin{align}
    \hypothesis(\tau) = \big\{ \nn \in \hypothesis : \bar{\medloss}_\text{phy}(\medout, \theta(\medmultimodalin)) \leq \tau 
    \big\},
\end{align}
where $\bar{\medloss}_\text{phy}(\medout, \theta(\medmultimodalin)) = \E_{(\medmultimodalin, \medout) \sim \P(\medmultimodalin, \medout)}[\medloss_\text{phy}(\medout, \theta(\medmultimodalin)) ]$ is the expected loss function. In the following theorem, we derive the impact of the physics-based loss function $\bar{\medloss}_\text{phy}$ on the data efficiency.
\begin{theorem} \label{theorem:1}
    Suppose $\exists \nn^* \in \hypothesis$ s.t. $\bar{\medloss}_\text{data}(\medout, \nn^*(\medmultimodalin)) = 0$ and $\bar{\medloss}_\text{phy}(\medout, \nn^*(\medmultimodalin)) = 0$. Then, $\forall \epsilon_0, \epsilon_1 \in (0, 1)$ with $\delta \in (0, 1)$, and $m$ samples from $\P(\medmultimodalin, \medout)$ are sufficient to achieve error $\epsilon_1$ with probability $1 - \delta$, where
    \begin{align}
        m \geq \frac{1}{\epsilon_1} \left[  
        \ln(|\hypothesis(\epsilon_0)|) + \ln\frac{1}{\delta}
        \right] 
    \end{align}
\end{theorem}
\begin{proof}
See Appendix \ref{proof: theorem1}.
\end{proof}
From Theorem \ref{theorem:1}, we can see that the physics-based loss function reduces the required amount of data $m$ to achieve error $\epsilon_1$. As $\epsilon_0$ decreases, the size of $\hypothesis(\epsilon_0)$ decreases, leading to a smaller amount of data to achieve error $\epsilon_1$. Hence,  $\medloss_\text{phy}$ regularizes the \ac{ML} model to be consistent with how channels propagate, thereby improving the data efficiency and robustness to domain shifts. Meanwhile, if we set $\epsilon_0$ as $\infty$, the size of $\hypothesis(\epsilon_0)$ becomes the whole possible space of $\hypothesis$, resulting in more required data. Here, we note that having $\nn_k^*$ with zero loss is a theoretical limit that may not be achievable in real-world scenarios. However, Theorem \ref{theorem:1} still shows that any $\nn_k$ that adheres to the designed $\medloss_{\text{phy}}$ can improve the data efficiency. In addition, we corroborate our theorem in simulations as shown in Fig. \ref{fig: agent1}.

\subsection{Collaborative Domain Adaptation} \label{subsec: collaborative_domain_adaptation}
In the second phase, we aim to solve problem \eqref{prob: adaptation} to find the minimal amount of data from $\P^{\text{test}}_k(\medmultimodalin, \medout)$ to adapt the model under domain shifts. During inference, each \ac{BS} is deployed to make network decisions in real-time. If the performance of \ac{BS}'s model degrades due to domain shifts, it should adapt its model to the current domain, using data samples from $\P^{\text{test}}_{k}(\medmultimodalin, \medout)$. However, the amount of available data is extremely limited during inference, and the adaptation must happen in few epochs to support the real-time inference. Hence, we propose a \emph{collaborative domain adaptation framework}, in which \ac{BS}s collaborate to guide under-performing \ac{BS}s to improve the data-efficiency and adaptation speed. Essentially, if a given \ac{BS} experiences a domain shift, it can request help from other \ac{BS}s to retrieve their models. Then, it aggregates the received models based on similarities between the current domain and other \ac{BS}s' domains. Hence, by using the knowledge of \ac{BS}s who have high domain similarity, the under-performing \ac{BS} can adapt to the current domain promptly with limited amount of data. 

\subsubsection{Measuring Domain Similarities}
We first study how to measure the domain similarity between two \ac{BS}s. The last layer of $\nn_k$ will be given by $\head_k\in \mathbb{R}^{\medoutputnum \times \dimfeature}$, where $\medoutputnum$ is the dimension of the network decision of $\nn_k$ and $\dimfeature$ is the dimension of feature vector $\feature \in \mathbb{R}^\dimfeature$. The output of $\nn_k$ can be given by $\medout_k = \langle \feature, \head_k \rangle $.
\ac{BS} $k$ encodes multi-modal input $\medmultimodalin_k$ into features $\feature_k$ using deterministic or probabilistic mapping (e.g., dropout) with $\nn_k$. Hence, we can know that the joint probability distribution $\P_k(\medmultimodalin, \medout, \feature)$ contains almost the same information between $\medmultimodalin_k$ and $\medout_k$ as $\P_k(\medmultimodalin, \medout)$ does. As such, we consider $\P_k(\medmultimodalin, \medout, \feature)$ as a domain of \ac{BS} $k$ hereinafter. 

By using Bayes rule, we have:
\begin{align}
    \P_k(\medmultimodalin, \medout, \feature) &= \frac{\P_k(\medout| \feature, \medmultimodalin)} {\P_k(\medout|\medmultimodalin)} \times \P_k(\feature | \medmultimodalin) \times \P_k(\medmultimodalin, \medout) \\
    &\approx \P_k(\feature|\medmultimodalin) \P_k(\medmultimodalin) \E_{\P_k(\feature|\medmultimodalin)} 
    \left[
    \P_k(\medout|\feature)
    \right] \label{domain_10} \\
    &= \P_k(\feature) \E_{\P_k(\feature|\medmultimodalin)} 
    \left[
    \P_k(\medout|\feature)
    \right] \\
    &\approx \P_k(\feature), \label{joint_domain}
\end{align}
where \eqref{domain_10} follows from the fact that $\medmultimodalin$ and $\feature$ provide the almost same information to $\medout$, and \eqref{joint_domain} results from the fact that $\medout$ is deterministic for given $\feature$ and the last layer $\head$. From \eqref{joint_domain}, we can see that a domain can be approximated as the probability distribution of the encoded multi-modal features $\feature_k$. We assume that $\P_k(\feature)$ follows a multivariate normal distribution \cite{tian2021exploring}. Each \ac{BS} $k$ tracks its domain $\P_k(\medmultimodalin, \medout, \feature)$  by calculating the mean $\mean_k$ and covariance $\covariance_k$ of $\P_k(\feature)$ in the current batch. We use the Wasserstein-2 distance $\distance(\cdot)$ to measure the similarity of the domains of two \ac{BS}s as it holds symmetry and triangle inequality. Then, the similarity of the domains of two \ac{BS}s can be given as
\begin{align}
	\distance(k, j) = \sqrt{||\mean_k - \mean_j ||^2_2 + \text{Tr}(\covariance_k + \covariance_k -2\sqrt{\covariance_k^{1/2} \covariance_j \covariance_k^{1/2})} }. \label{distance}
\end{align} 
From \eqref{distance}, we can measure how similar their domains are. If \ac{BS} $k$ experiences an unseen domain, \ac{BS}s with similar domains can have knowledge that is helpful for adapting to the current domain. Next, we present how to aggregate models of other \ac{BS}s based on domain similarities. 

\subsubsection{Domain Similarity Aware Aggregation}
Under domain shifts, \ac{BS} $k$ first requests help to other \ac{BS}s to retrieve their model parameters $\{\nn_i\}, \forall i \neq k$. Then, \ac{BS} $k$ calculates $\distance(k, i)$, $\forall i \neq k$ using the received models and data $\medmultimodalin_k'$ from the current domain. Lastly, it generates a new model by aggregating the received models based on the measured domain similarities as follows
\begin{align}
    &\nn_k \leftarrow \sum_{i\neq k} \gamma_i \nn_i, \ \text{where} \ \gamma_i = \text{softmax}
    \left(
    \frac{1}{ \distance(k,  i) }
    \right).  \label{aggregation}
\end{align}
We assign higher weights to models that show higher domain similarity with the current data samples $\medmultimodalin_k'$. After the aggregation, \ac{BS} $k$ adapts its model to $\medmultimodalin_k'$ for a few epochs. This collaborative domain adaptation improves the data efficiency and adaptation speed because \ac{BS} $k$ initializes its model using other \ac{BS}s' models who have high domain similarities. Hence, the proposed framework allows the adaptation with significantly limited amount of data during inference. The overall pipeline is summarized in Algorithm \ref{algo:1}.

The proposed collaborative domain adaptation has complexity of $\mathcal{O}(\numagents)$ because an under-performing \ac{BS} needs to retrieve model parameters of other \ac{BS}s and calculate domain similarities.

\begin{algorithm}[t!]  \footnotesize
\caption{The overall pipeline for the physics-based training and collaborative domain adaptation. }  \label{algo:1}
\KwInput{Training dataset $\{\medmultimodalin_k, \medout_k \}$, and model $\nn_k$, and data from unseen domain $\medmultimodalin_k', \forall k$.}
\textbf{Training}:  \\
\quad \For {each agent $k$ with $\{\medmultimodalin_k, \medout_k \}$ }{
            Update $\nn_k$ using \eqref{physics_loss}; \\
        }    
\textbf{Inference}: \\
\quad \For {each \ac{BS} $k$ with $\medmultimodalin_k'$ (in parallel)}
{         
            Calculate $\mean_k$ and $\covariance_k$; \\
            \If{Domain shift occurs}
            {      	Retrieve models of other \ac{BS}S; \\
                	Calculate domain similarity $\gamma$ using \eqref{distance}; \\
                	Aggregate the retrieved models using \eqref{aggregation};\\
                	Adapt $\nn_k$ using $\medmultimodalin_k'$;
            		}       
}

    \vspace{-0.1cm}
\end{algorithm}

\section{Use case: Multi-modality based \ac{mmWave} RSS prediction} \label{sec: med_use_case}

\begin{figure*}[t!]
	\begin{subfigure}[t]{0.43\textwidth}
			\begin{center}   				%
				\includegraphics[width=\textwidth]{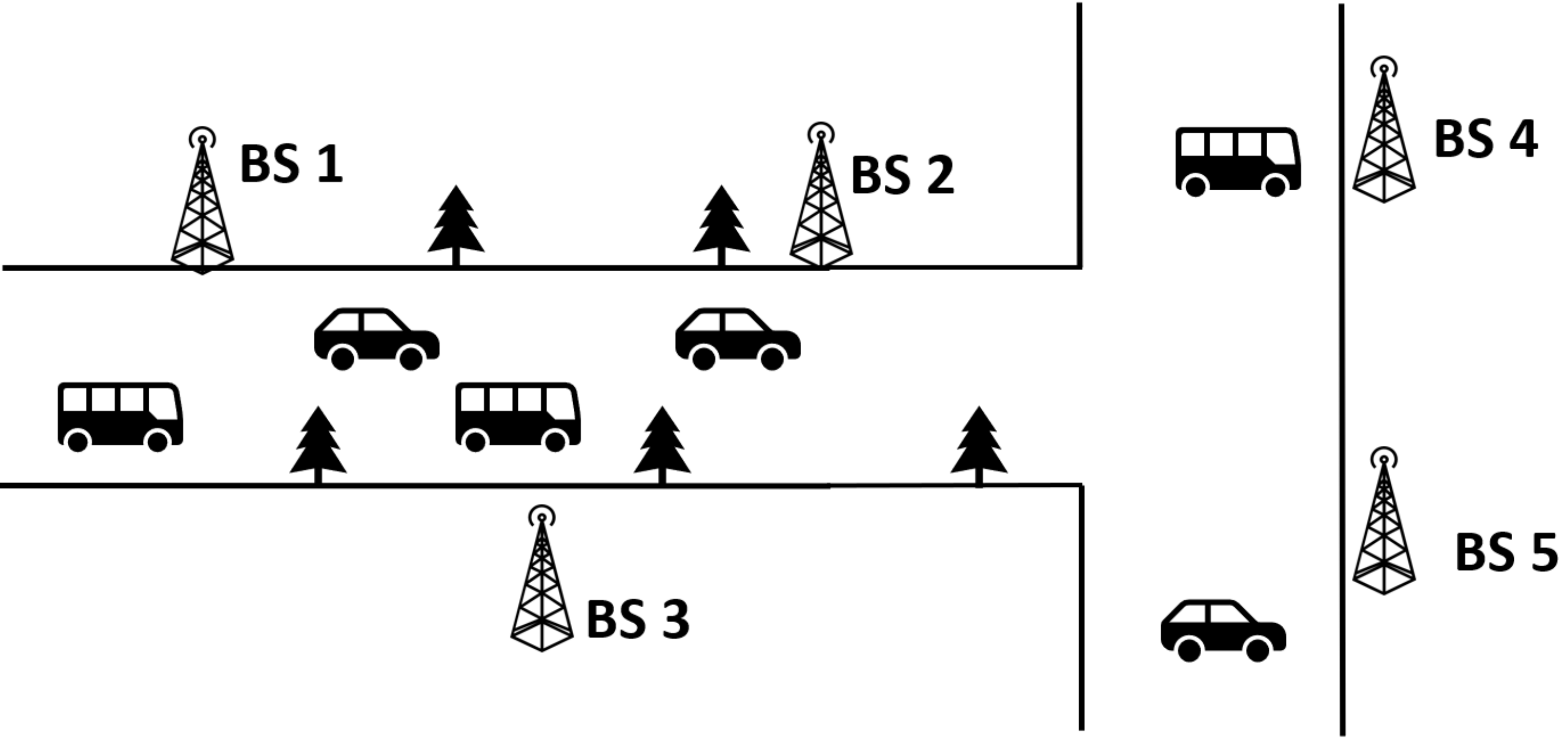}
		\end{center}
		\caption{An illustration of the use case system model.}
		\label{fig: use_case_1}
	\end{subfigure}\hfill
	\begin{subfigure}[t]{0.43\textwidth}
        \begin{center}   
		\includegraphics[width=\textwidth]{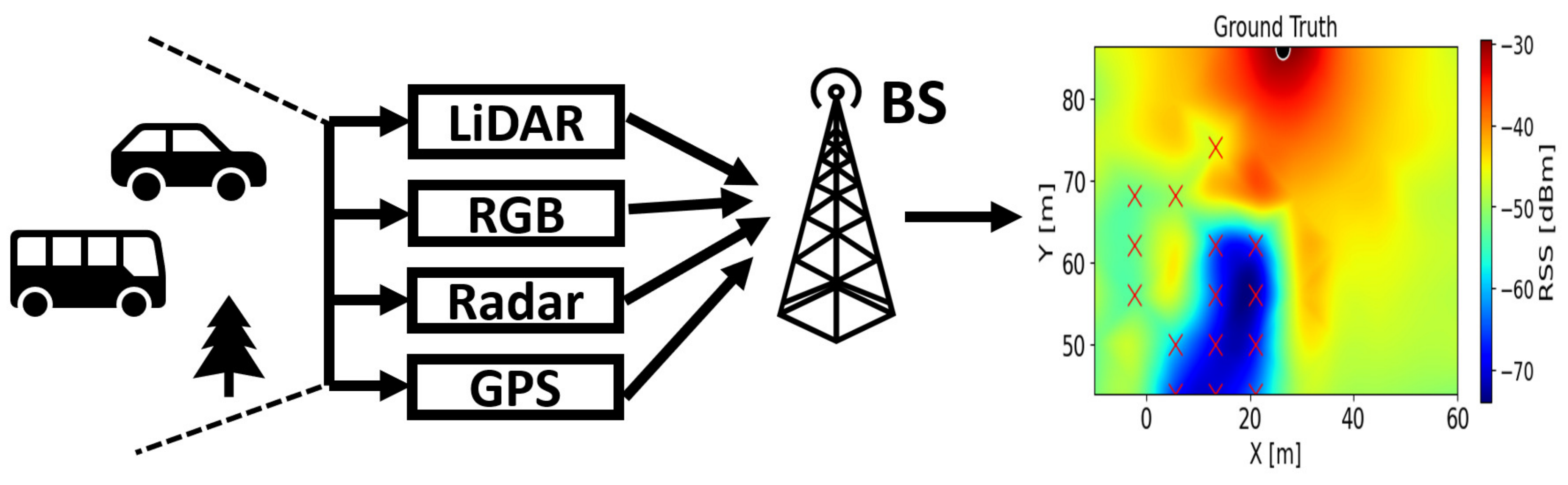}
		\end{center}
		\caption{An illustration of multi-modal based RSS map prediction.}
		\label{fig: use_case_2}
	\end{subfigure}\hfill
    \caption{An illustration of the considered use case and multi-modal based \ac{RSS} prediction.}
    \label{fig: use_case}
    \vspace{-0.3cm}
\end{figure*}

To validate the proposed frameworks in practical wireless networks scenarios, we present \ac{mmWave} \ac{RSS} prediction as a use case because it can leverage multi-modal sensors to capture surrounding objectives.

\subsection{System Overview}
We consider multi-modal sensor (camera, LiDAR, radar, GPS) based \ac{mmWave} \ac{RSS} map prediction in urban vehicular scenarios as shown in Fig. \ref{fig: use_case}. We construct a \ac{mmWave} radio environment map that can capture the dynamics of surrounding environment from multi-modal sensors. Hence, the constructed map can potentially be used to guide networking algorithms in dynamic scenarios such as resource allocation or network slicing.  
Based on a captured multi-modal input $\medmultimodalin_k$, each \ac{BS} $k$ predicts the current \ac{RSS} values $\medout_k$ over a set of receiver grids $\mathcal{\medoutputnum}$ of $\medoutputnum = \medoutputnum_x \times \medoutputnum_y$ with height $\rxheight$ of its coverage area using its model $\nn_k$. The \ac{RSS} of a receiver $(i, j)$, where $1\leq i \leq \medoutputnum_x$ and $1 \leq j \leq \medoutputnum_y$, significantly depends on whether it has an \ac{LoS} path or not. If $(i, j)$ has an \ac{LoS} path, its \ac{RSS} is largely governed by the distance to the serving \ac{BS} and first-order reflections \cite{li2024map, jaeckel2017explicit} as follows
\begin{align}
    \rss_{ij}  = \rss_{ij}^{\text{LoS}} \hspace{-0.75mm} +  \rss_{ij}^{\text{reflection}}, \ \text{if $(i,j)$ has an LoS path}, \label{rss_los}
\end{align}
where $\rss_{ij}^{\text{LoS}}$ is the \ac{RSS} from a direct beam between $(i, j)$ and the serving \ac{BS} and $\rss_{ij}^{\text{reflection}}$ is the \ac{RSS} from the reflected beams from surrounding obstacles and grounds. We model $\rss_{ij}^{\text{LoS}}$ using the 3GPP-Umi \ac{LoS} path loss modeling \cite{3gpp} as follows 
\begin{align}
    \pathloss_{ij}^\text{LoS} = 32.4 + 17.3\log_{10}(d_{ij}) + 20\log_{10}(\carrier)+ \shadow,
\end{align}
where $d_{ij}$ is the distance between  $(i, j)$ and its serving \ac{BS}, $\carrier$ is the carrier frequency normalized by the unit of GHz, and $\shadow$ is the shadowing factor. Then, $\rss_{ij}^{\text{LoS}}$ can be given as
\begin{align}
    \rss_{ij}^{\text{LoS}} = \txpower + \txgain + \rxgain - \pathloss_{ij}^{\text{LoS}},
\end{align}
where $\txpower$ is the transmission power and $\txgain$ and $\rxgain$ are the directional transmit and receiver antenna gains, respectively. If there is no \ac{LoS} path, the \ac{RSS} of $(i, j)$ will be highly dependent on the location of the receiver and the attenuation and diffraction values of blockages. Hence, for the \ac{NLoS} case, we model the \ac{RSS} as  $\rss_{ij}^{\text{LoS}}$ attenuated by blockages:
\begin{align}
    \rss_{ij} \hspace{-0.mm} = \hspace{-0.mm} \rss_{ij}^{\text{LoS}} \hspace{-0.mm} - \hspace{-0.mm} \rss_{ij}^{\text{blockage}} \hspace{-0.mm} , \ \text{if $\hspace{-0.mm} (i,j \hspace{-0.mm})$ has no LoS path}, \label{rss_nlos}
\end{align}
where $\rss_{ij}^{\text{blockage}}$ is the attenuated \ac{RSS} due to blockage.

We assume that each \ac{BS} knows the position of vehicles in its coverage area as estimated by basic safety messages broadcast by the vehicles \cite{bsm}. Next, we apply the physics-based training to the proposed use case.

\subsection{Physics-Based Training}
We first specify the $\medloss_{\text{total}}$ term as defined in \eqref{physics_loss}. 
For $\medloss_{\text{data}}(\medout_k, \nn_k(\medmultimodalin_k))$, we use an \ac{MSE} loss function between predicted \ac{RSS} and true \ac{RSS} as $\medloss_{\text{data}}(\medout_k, \nn_k(\medmultimodalin_k)) = \frac{1}{\medoutputnum} \sum_{(i,j) \in \mathcal{N}}
\left(
\widehat{\rss}_{k, ij} - \rss_{k, ij}
\right)^2$.
For $\medloss_{\text{phy}}(\medout_k, \nn_k(\medmultimodalin_k))$, we divide it into the case when receiver $(i,j), 1 \leq i \leq \medoutputnum_x$ and $1 \leq j \leq \medoutputnum_y $ has a \ac{LoS} path or only has \ac{NLoS} paths. 
We design $ \medloss_{\text{LoS}}(\nn_k(\medmultimodalin_k))$ to align the predictions of $\nn_k$ with the physical law of \ac{LoS} scenarios as follows:
\begin{align}
    \medloss_{\text{LoS}}(\nn_k(\medmultimodalin_k)) &= \frac{1}{|\mathcal{\medoutputnum}_{\text{LoS}}|} \sum_{(i,j) \in \mathcal{\medoutputnum}_{\text{LoS}}}
    \bigg[
    \rss_{k, ij}^{\text{blockage}} \ka
    &\quad + \max
    \left(
    0, \rss_{k ,ij}^{\text{reflection}} - \reflection_{k, ij}
    \right)
    \bigg], \label{los_loss}
\end{align}
where $\mathcal{\medoutputnum}_{\text{LoS}}$ is a set of receivers that have a \ac{LoS} path and $\reflection_{k, ij}$ is the bound of reflection gain $\rss_{k ,ij}^{\text{reflection}}$. Intuitively, for receivers with \ac{LoS} paths, the \ac{RSS} reduction term $\rss_{k, ij}^{\text{blockage}}$ should be minimized as there should be no blockage. We use a hinge loss $ \max
\left(
0, \rss_{k ,ij}^{\text{reflection}} - \reflection_{k, ij}
\right)$ to penalize too large \ac{RSS} gains from reflections. Here, $\reflection_{k, ij}$ is also a predicted value to bound $\rss_{k ,ij}^{\text{reflection}}$ based on the captured multi-modal input $\medmultimodalin_k$ on the location of receiver $(i, j)$. Similarly, we define $\medloss_{\text{NLoS}}(\nn_k(\medmultimodalin_k))$ to allow $\nn_k$ to learn the channel propagation of \ac{NLoS} cases as follows
\begin{align}
    \medloss_{\text{NLoS}}(\nn_k(\medmultimodalin_k)) &= \frac{1}{|\mathcal{\medoutputnum}_{\text{NLoS}}|} \sum_{(i,j) \in \mathcal{\medoutputnum}_{\text{NLoS}}}
    \bigg[
    \rss_{k, ij}^{\text{reflection}} \ka
    &\quad + \max
    \left(
    0, \blockage_{k, ij} - \rss_{k ,ij}^{\text{blockage}}
    \right)
    \bigg], \label{nlos_loss}
\end{align}
where $\mathcal{\medoutputnum}_{\text{NLoS}}$ is a set of receivers that only have \ac{NLoS} paths and $\blockage_{k, ij}$ is the bound of the \ac{RSS} reduction from blockages $\rss_{k, ij}^\text{blockage}$. We minimize $\rss^\text{reflection}$ to penalize large reflection gains for receivers in $\mathcal{\medoutputnum}_{\text{NLoS}}$. For \ac{mmWave} signals, blockages often decrease \ac{RSS} significantly. Hence, we train $\blockage_{k, ij}$ to penalize too small \ac{RSS} attenuation in $\mathcal{\medoutputnum}_{\text{NLoS}}$. Note that $\blockage_{k, ij}$ is also a predicted value based on $\medmultimodalin$.  The labels of  $\mathcal{\medoutputnum}_{\text{LoS}}$ and $\mathcal{\medoutputnum}_{\text{NLoS}}$ are available from ray-tracing results when generating ground-truth \ac{RSS}. We provide more detailed information on how we generate realistic datasets in Sec. \ref{sec: med_data_generation}.

\begin{figure}[!t]
    \centering
    \includegraphics[width=1.0\columnwidth]{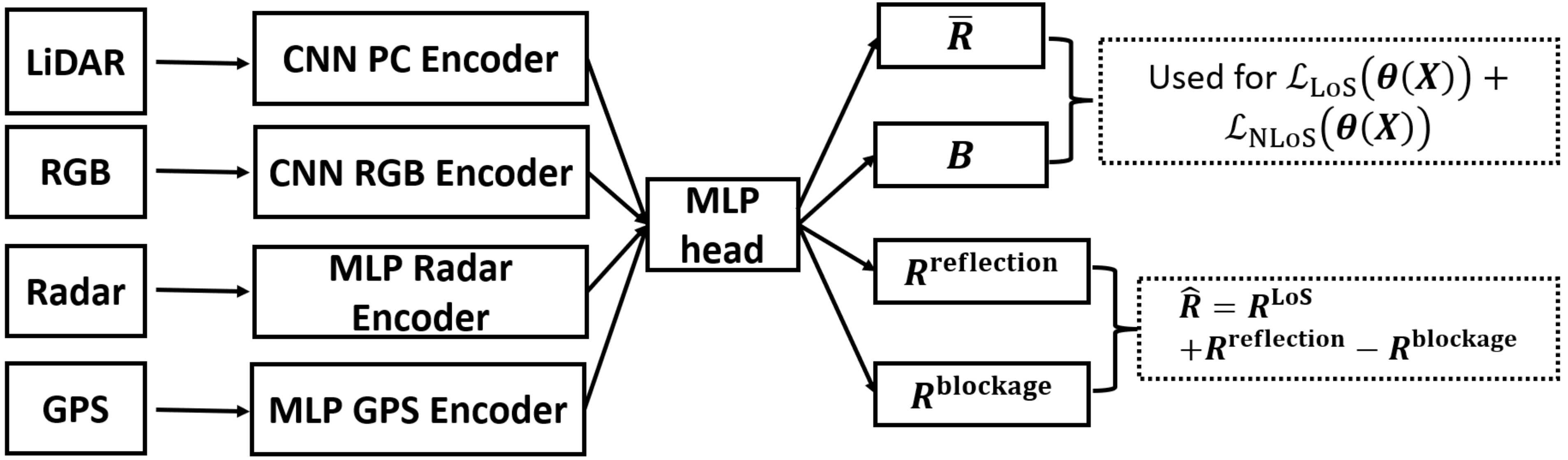}
    \captionsetup {singlelinecheck = false}    
    \vspace{-0.1cm}
    \caption{An illustration of the multi-modal model architecture.}
    \label{fig:model}    
    \vspace{-0.1cm}	
\end{figure}

Now, we present the architecture of our \ac{ML} model in Fig. \ref{fig:model} to process multi-modal input $\medmultimodalin$. The designed \ac{ML} outputs $\widehat{\rss}_{k, ij}, \reflection_{k, ij},$ and $\blockage_{k, ij}$ for $\forall k \in \mathcal{\numagents}$ and $\forall (i,j), 1 \leq i \leq \medoutputnum_x$ and $1 \leq j \leq \medoutputnum_y.$ To capture unique pattern in each modality, we design the corresponding encoders to process raw modality data. We utilize a small \ac{CNN} encoder for RGB inputs. For LiDAR, we first project point clouds into 2D dimension with bird's-eye view and use a \ac{CNN} encoder to extract transformed 3D features. We use the highest point sampling to pre-process radar inputs and encode them with \ac{MLP} layers. Similarly, an \ac{MLP} encoder is used to process GPS inputs since they are low dimension vectors. We provide more detail about the pre-processing in Section. \ref{sec: med_experiments}. We concatenate the output of each encoder and forward it to the \ac{MLP} head for the fusion. The \ac{MLP} head predicts values $\rss_{k, ij}^\text{reflection}, \rss_{k, ij}^\text{blockage}, \reflection_{k, ij},$ and $\blockage_{k, ij},$ $\forall (i,j), 1 \leq i \leq \medoutputnum_x$ and $1 \leq j \leq \medoutputnum_y$, using the shared encoded output from each encoder. We use all the predicted values during the training to calculate the physics-guided losses \eqref{los_loss} and \eqref{nlos_loss}. After training, each \ac{BS} $k$ predicts $\widehat{\rss}_{k, ij} = \rss^\text{LoS}_{k, ij} + \rss^\text{reflection}_{k, ij} - \rss^\text{blockage}_{k, ij} $ discarding $\reflection_{k, ij}$ and $\blockage_{k, ij}$. More detailed information of the pre-processing of each modality data and architecture is provided in Sec. \ref{sec: med_experiments}.

As described in Sec. \ref{subsec: collaborative_domain_adaptation}, after training, each \ac{BS} $k$ tracks the statistics of $\mean_k$ and $\covariance_k$ under its current domain. If \ac{BS} $k$ experiences a domain shift, it requests help from other agents to retrieve their models. Here, the communication overhead is negligible because of the throughput of typical backhauls of \ac{BS}s. For instance, the throughput of V-band can up to 10 Gbps \cite{tezergil2022wireless}. Subsequently, it measures the domain similarities of other \ac{BS}s and aggregates their models based on \eqref{aggregation}. Then, it adapts its model to the new samples for a few epochs. Next, we present how to generate multi-modal training data for the \ac{mmWave} \ac{RSS} prediction. 

\section{Multi-modal dataset generation for \ac{mmWave} RSS prediction} \label{sec: med_data_generation}
To validate the idea of the proposed framework with the use case, we need a multi-modal dataset that can induce covariate/concept shifts to multiple \ac{BS}s in dynamic wireless environments. To the best of our knowledge, there is no a such public dataset for multi-modal \ac{mmWave} \ac{RSS} prediction. We now present a multi-modal simulation framework to generate our training data. Our data generation framework is based on \cite{park2025resource}, which predicts \ac{mmWave} beamforming vectors, with the following changes. Firstly, we change the task as \ac{mmWave} \ac{RSS} prediction in the surrounding area of \ac{BS}s. We also include multiple \ac{BS}s for collaborative domain adaptation to tackle domain shifts. Meanwhile, \cite{park2025resource} is used for a single \ac{BS} without considering domain shifts.

The workflow of our simulation framework is illustrated in Fig. \ref{fig:data_generation}. We use the autonomous vehicular simulator CARLA \cite{dosovitskiy2017carla} to generate multi-modal sensing information and a virtual urban environment. Based on the generated urban environment and vehicles, we use MATLAB to do wireless communication experiments. Our framework has four main steps: CARLA environment setup, multi-modal sensing data generation, map reconstruction, and wireless communication experiments based on the generated environments. We integrate CARLA's multi-modal sensing ability and vehicular simulations into MATLAB's communication toolboxes. This approach provides flexibility to make domain shifts to vehicular environments  in CARLA and to see their impact on wireless environments in MATLAB.

\subsubsection{CARLA Environment Setup} We generate multiple vehicles and \ac{BS}s within the CARLA environment. We use CARLA's traffic control and navigation functions for autonomous movements of vehicles and traffic rules. Our environment setup enables dynamic blockage patterns and multi-modal sensing conditions to generate various domain shifts. 
\begin{itemize}
    \item \emph{\ac{BS} placements}: We place \acp{BS} at fixed points as roadside units.

    \item \emph{Vehicle placements}: We spawn vehicles in random locations around the generated \acp{BS}. Spawned vehicles follow pre-defined traffic rules and speed limits. 

    \item \emph{Environmental dynamics}: We control brightness and traffic density to simulate different covariate and concept shifts (e.g., day/night or heavier traffic). 
\end{itemize}

\begin{figure}[!t]
    \centering
    \includegraphics[width=1.0\columnwidth]{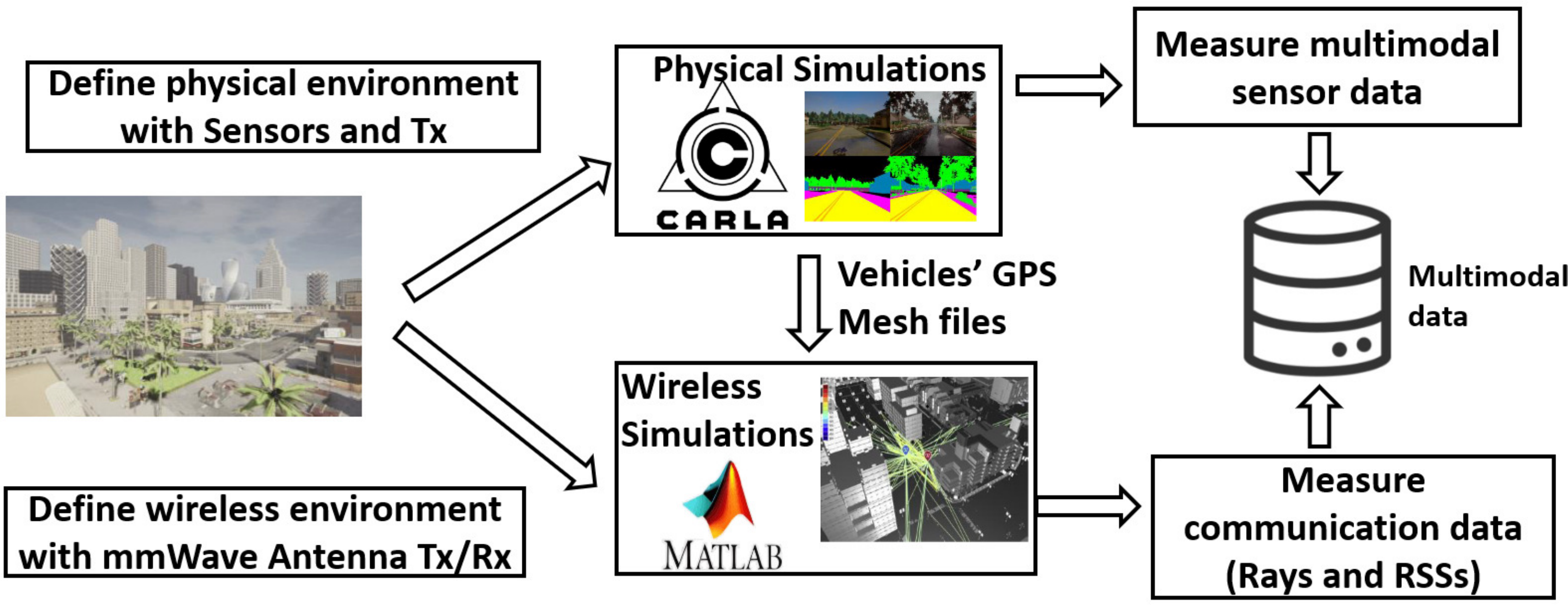}
    \captionsetup {singlelinecheck = false}

    \vspace{0.0cm}
    \caption{An illustration of the multi-modal simulation framework based on CARLA and MATLAB.}
    \label{fig:data_generation}

    \vspace{-0.0cm}	
\end{figure}

\subsubsection{Generating Multi-modal Sensing Data} At each \ac{BS}, we instantiate the following sensor modalities:
\begin{itemize}
    \item An RGB camera that captures the current frame to provide visual context such as blockage detection and object classification. 
    \item A LiDAR sensor that provides distance, object shape, and blocked region in 3D point clouds. 
    \item A radar that measures distance and velocity information of surrounding objects.
    \item A GPS that records vehicle coordinates and velocities. 
\end{itemize}
Based on multi-modal data, the \acp{BS} can capture the details of dynamic environments such as movements of vehicles, surrounding obstacles, and buildings. We synchronize all sensing data in a time frame. 

\subsubsection{Map Reconstruction} To do wireless communication experiments in MATLAB, we need to import the generated CARLA environments into MATLAB. We follow the steps in \cite{park2025resource} to make the CARLA environments compatible with MATLAB using Blender API \cite{conlan2017blender}. We first export the CARLA environment into a ply 3D file format. We then use the Blender API to script the conversion process to preserve the geometry, coordinates, and scales of the environments. The output of the Blender API is imported into MATLAB as a mesh file, which is consistent with the original CARLA environment.   

\subsubsection{Wireless Communication Simulation} We import the generated 3D environment file into MATLAB to perform ray-tracing to capture \ac{mmWave} channel propagation in the current environment. We set each \ac{BS} as a transmitter and locate multiple receivers over the grids of the coverage area of the \ac{BS}. We trace signal paths, path loss, reflection, diffraction, and \ac{LoS} existence of each receiver. The \ac{RSS} of each receiver is calculated for each received ray considering reflections, scattering, and diffraction. To enable dynamic environmental change, updated vehicle locations from CARLA are used to simulate ray-tracing. 

We combine the multi-modal sensor data from CARLA and the corresponding ray-tracing results including \ac{RSS} and \ac{LoS} labels. Specifically, our multi-modal dataset consists of:
\begin{itemize}
    \item Multi-modal Sensor Input: RGB, LiDAR, Radar, and GPS.
    \item Ray-tracing Results: \ac{RSS}, path loss, delay spread, angle of arrival, angle of departure, and \ac{LoS} existence of each receiver.
    \item Time and positional information: Time frame, vehicle IDs, and coordinates of vehicles, \acp{BS}, and receivers. 
\end{itemize}
We use the constructed dataset for the proposed physics-based training in Sec. \ref{subsec: physics_based_training} and collaborative domain adaptation in Sec. \ref{subsec: collaborative_domain_adaptation} to validate the proposed framework. 

In summary, our multi-modal data generation framework integrates autonomous vehicles and wireless communication simulations. We essentially generated dynamic urban environments to capture dynamic blockage and object patterns caused by moving vehicles. Then, we imported the CARLA environments to MATLAB to do ray-tracing. Our approach captures how the dynamics of the environment affect the \ac{mmWave} channel propagation. We use our multi-modal data to simulate the proposed use case in Sec. \ref{sec: med_use_case}. 

\section{Simulation Results and Analysis} \label{sec: med_experiments}
We generate our dataset using the `Town 10' map in CARLA to simulate an urban environment with up to 50 vehicles. We deploy five \acp{BS} along the map. We provide the example RGB frames of \acp{BS} 1 to 4 in Fig. \ref{fig: agents}. We set a 300 ms time interval between each time frame. We generate around 4K training datasets, 1K validation datasets called VAL-1, and another 2K validation datasets called VAL-2 for each \ac{BS}. For VAL-1, we induce new traffic patterns by making unseen blockages or obstacle patterns as shown in Fig. \ref{fig:agent1_4_con}, thereby inducing covariate and concept shifts. As such, new blockage patterns induce unseen mapping between $\medmultimodalin$ and $\medout$, and the distribution of $\medmultimodalin$ also changes. For \acp{BS} 1, 2, and 3, we induce new obstacle patterns by placing large vehicles at different locations, leading to concept shift-dominated domain shifts. For \acp{BS} 4 and 5, we induce stronger domain shifts by generating unseen large vehicles. For VAL-2, we decrease the brightness level by 25\% and inject zero mean Gaussian noises $\mathcal{N}(0, 1), \mathcal{N}(0, 0.5^2),$ and $\mathcal{N}(0, 0.25^2)$ to Radar, GPS, and LiDAR input data, respectively. For VAL-2, the induced domain shifts only change the distribution of inputs  $\P[\medmultimodalin]$ while $\P[\medout|\medmultimodalin]$ does not change. For the wireless communication simulations, we set the height of each \ac{BS} to 5.5 m and $\txpower = 25$ dBm \cite{ali2020passive, xu2022computer}. We assume each \ac{BS} is equipped with $4\times4$ UPA antennas. Receivers are distributed evenly over $\medoutputnum = 8\times8$ grids over $80$ m $\times$ $40$ m coverage area per each \ac{BS} \cite{narayanan2020lumos5g} to match the field-of-view of its camera. We assume receivers have a single antenna with $\rxheight = 1.5$ m.

\begin{figure*}[t!]
	\begin{subfigure}[t]{0.45\textwidth}
			\begin{center}   				%
				\includegraphics[width=\textwidth]{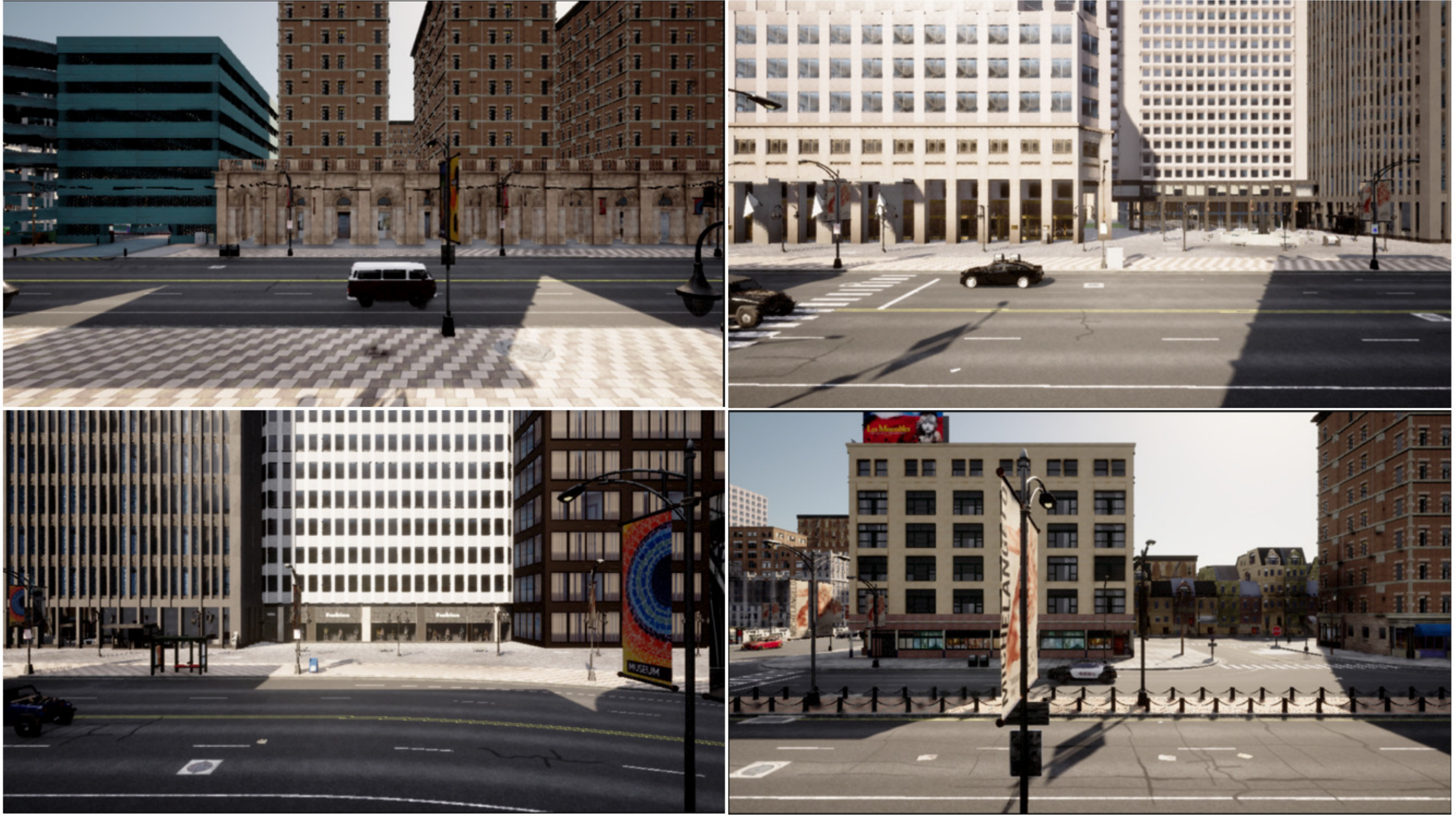}
		\end{center}
		\caption{Example RGB samples of \ac{BS} 1 to 4 from training datasets.}
		\label{fig: agent1_4}
	\end{subfigure}\hfill
	\begin{subfigure}[t]{0.45\textwidth}
        \begin{center}   
		\includegraphics[width=\textwidth]{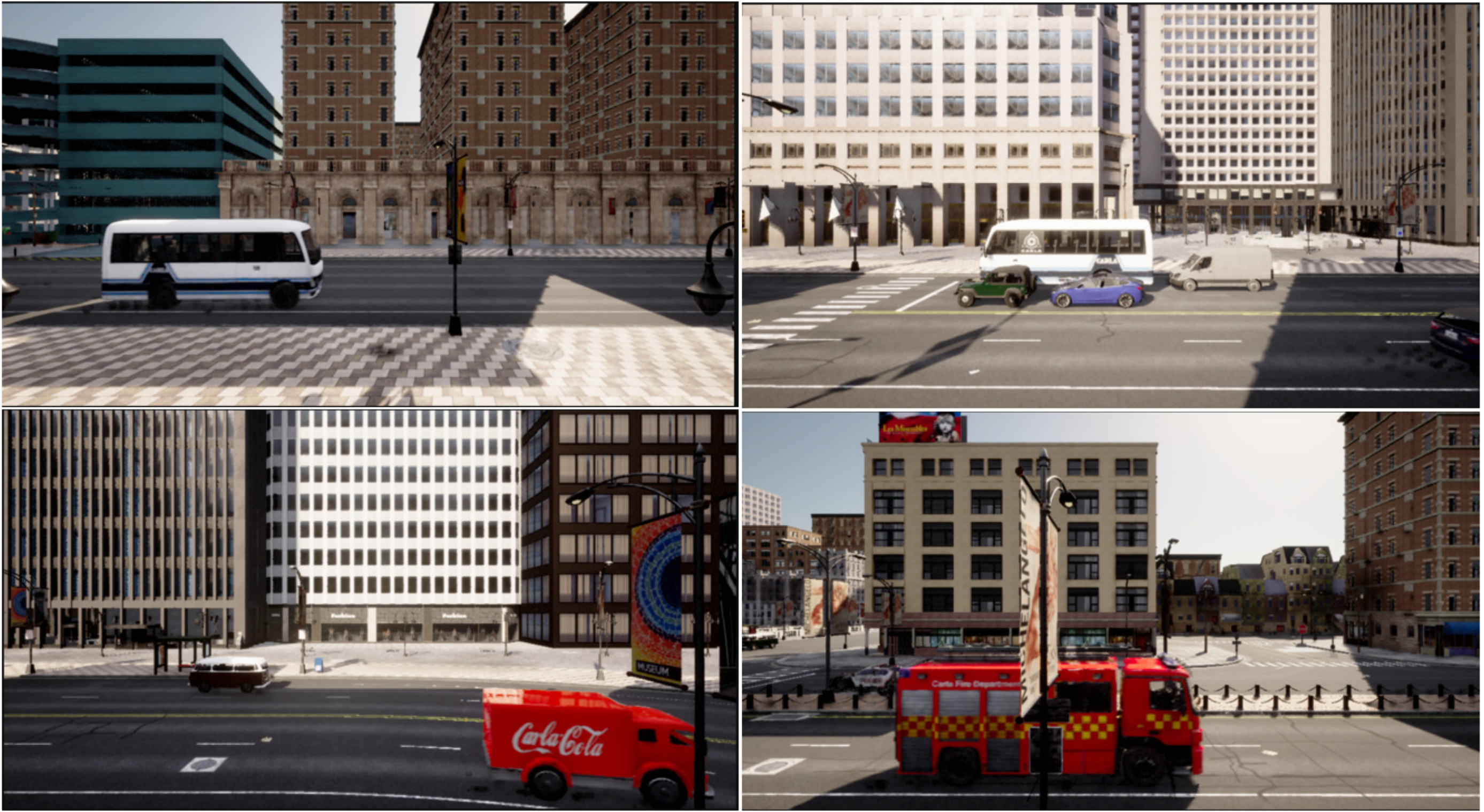}
		\end{center}
		\caption{Example RGB samples of \ac{BS} 1 to 4 from VAL-1 datasets.}
		\label{fig:agent1_4_con}
	\end{subfigure}\hfill
    \caption{Example RGB samples of \ac{BS} 1, 2, 3, and 4}
    \label{fig: agents}
\end{figure*}

We first pre-process each modality and forward it to the corresponding encoder as follows
\begin{itemize}
    \item \emph{LiDAR}: We project raw LiDAR point clouds [x, y, z, intensity] onto a 2D plane. This bird's-eye view transformation provides the locations and positions of surrounding objects while filtering out ground reflections \cite{park2025resource}. We use four convolutional layers and one \ac{MLP} layer to encode the pre-processed point clouds input. 

    \item \emph{Radar}: We use the highest point sampling to pre-process radar input [velocity, azimuth, altitude, depth]. Hence, we down-sample radar inputs to a fixed size. We use three \ac{MLP} layers with the max-pooling to encode the pre-processed radar inputs. 

    \item \emph{GPS}: We record GPS inputs as a csv format. We use three \ac{MLP} layers with the max-pool to encode GPS inputs. 

    \item \emph{RGB}: Camera images are first normalized and scaled to $256 \times 256$ resolution. We use a pre-trained ResNet-18 to encode RGB inputs. 
\end{itemize}
We concatenate all the encoded modality inputs and learnable receiver embedding to differentiate each receiver. We use three \ac{MLP} layers as our head for the prediction. 

We first train a model of each \ac{BS} using its training dataset following \eqref{physics_loss} and test on its VAL-1 and VAL-2 datasets. We average all results over five random seeds. We implement our training in PyTorch. We use 50 epochs with a batch size of 64 and set the initial learning rate as $1.0 \times 10^{-4}$ with the learning rate decay at (10, 30) epochs. We set $\phycoeff = 0.5$ for our physics-based loss. We consider three baselines for comparison purposes. \textbf{Baseline 1} predicts \ac{RSS} values directly without utilizing any physics prior information \eqref{rss_los}-\eqref{nlos_loss}. As such, Baseline 1 is trained only with \ac{MSE} loss, which is widely used for regression tasks \cite{jiao2025addressing, yao2022c, narayanan2020lumos5g}. \textbf{Baseline 2} \cite{seretis2022toward} penalizes \ac{LoS} and \ac{NLoS} receiver points differently when calculating \ac{MSE} loss function, by giving more weights to \ac{NLoS} points. We scale the loss of \ac{NLoS} points by 20\% more. \textbf{Baseline 3} \cite{li2024map} uses 3GPP path loss formulas for \ac{RSS} of \ac{LoS} points while using regression for \ac{NLoS} for input features.     We report \ac{MAE} and \ac{RMSE} for performance metrics.

\begin{figure*}[t!]
	\begin{subfigure}[t]{0.45\textwidth}
			\begin{center}   				%
				\includegraphics[width=\textwidth]{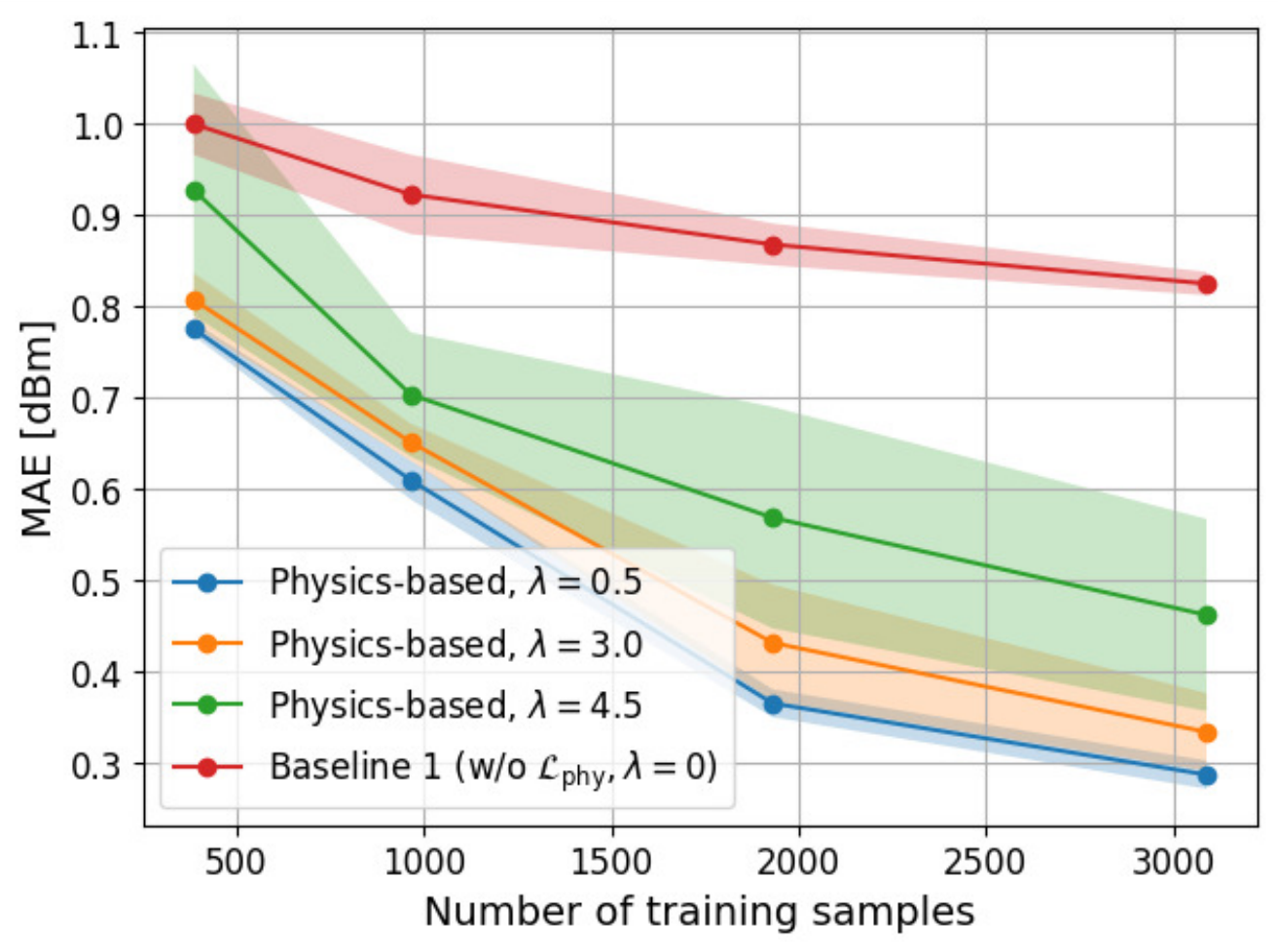}
		\end{center}
		\caption{\ac{MAE} on VAL-1 with respect to the amount of training data of \ac{BS} 1.}
		\label{fig: agent1_MAE}
	\end{subfigure}\hfill
	\begin{subfigure}[t]{0.45\textwidth}
        \begin{center}   
		\includegraphics[width=\textwidth]{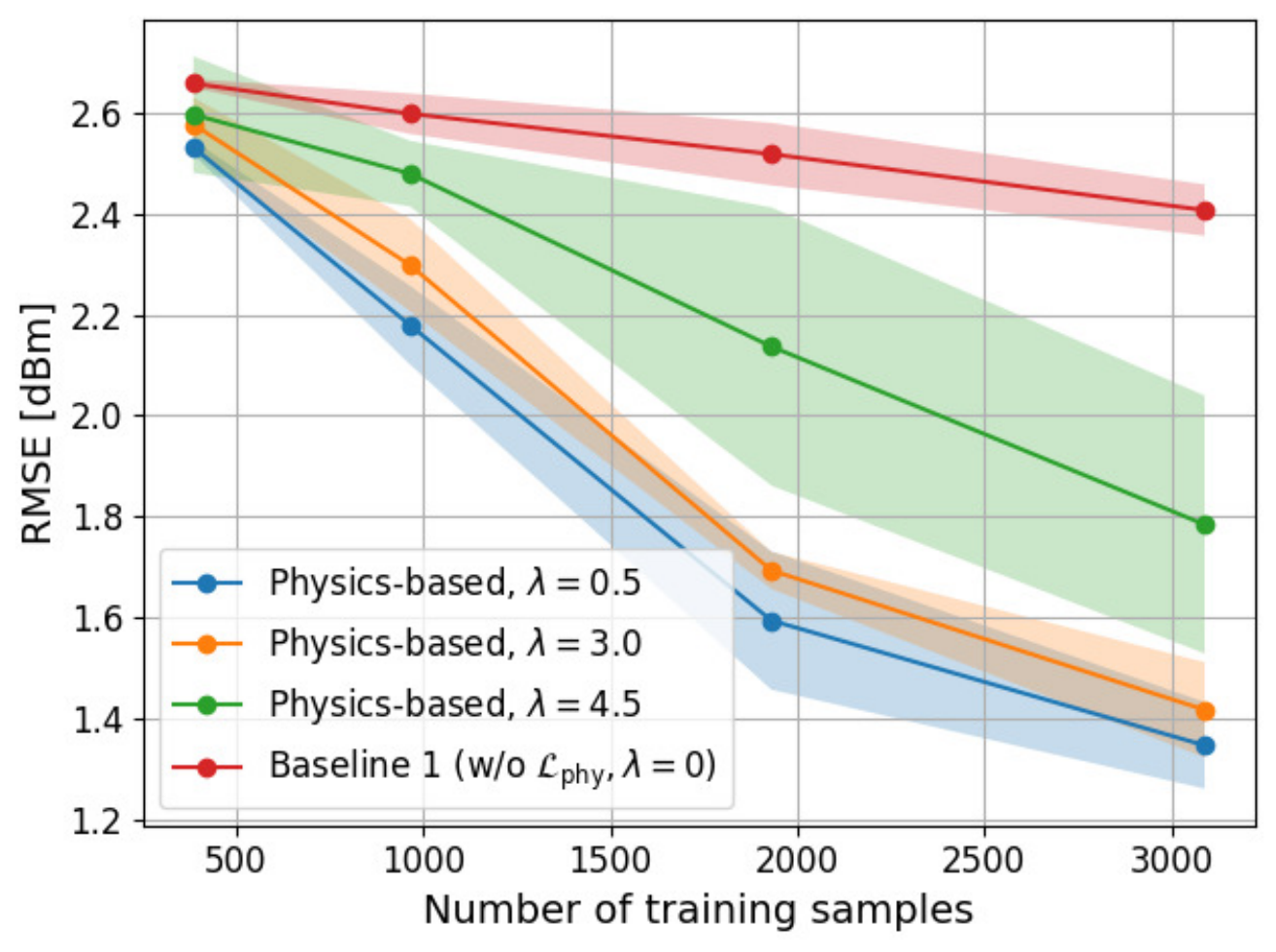}
		\end{center}
		\caption{\ac{RMSE} on VAL-1 with respect to the amount of training data of \ac{BS} 1.}
		\label{fig: agent1_RMSE}
	\end{subfigure}\hfill
    \caption{Performance of \ac{BS} 1 with respect to the amount of training data.}
    \label{fig: agent1}
\end{figure*}

Figure \ref{fig: agent1} shows the performance of the trained model of \ac{BS} 1 for an increasing amount of training data on its VAL-1 dataset. As shown in Fig. \ref{fig: agent1_MAE}, the VAL-1 induces concept shift-dominated domain shifts due to unfamiliar blockage and obstacle patterns from the training dataset. We can see that our physics-based model outperforms Baseline 1 in terms of the generalization performance and data efficiency. For the \ac{MAE}, the proposed model already achieves around 0.8 dBm of MAE with only 12.5\% of the raining data used by Baseline 1. Compared to Baseline 1, our physics-based model achieves faster improvement for increasing the amount of the training data. The physics-based loss function limits the available hypothesis space to a set that follows the designed physics law during training. This reduces the possible search space, and makes model parameters generalize robustly across unseen domains. Meanwhile, as $\phycoeff$ increases, the performance decreases since we regularize $\nn$ too strongly with $\medloss_{\text{phy}}$. Hence, the available search space becomes too limited for $\medloss_{\text{data}}$, leading to sub-optimal performance. Therefore, Figure \ref{fig: agent1} corroborates Theorem \ref{theorem:1}. 

For the \ac{RMSE}, we observe the same trend as the \ac{MAE} as shown in Fig. \ref{fig: agent1_RMSE}. The proposed model outperforms Baseline 1 with respect to the generalization performance and data efficiency. We can see that our model achieves around 2.2 dBm of \ac{RMSE} using only 31\% of the training data used by Baseline 1. By the definition of \ac{RMSE}, it penalizes large tails of \ac{MSE} values. In our dataset, receivers with only \ac{NLoS} paths have significantly low \ac{RSS} compared to \ac{LoS} receivers. Hence, we can see that the proposed model predicts \ac{RSS} of \ac{NLoS} areas more robustly than Baseline 1.

\begin{figure*}[t!]
	\begin{subfigure}[t]{0.45\textwidth}
		\centering   				%
		\includegraphics[width=\textwidth]{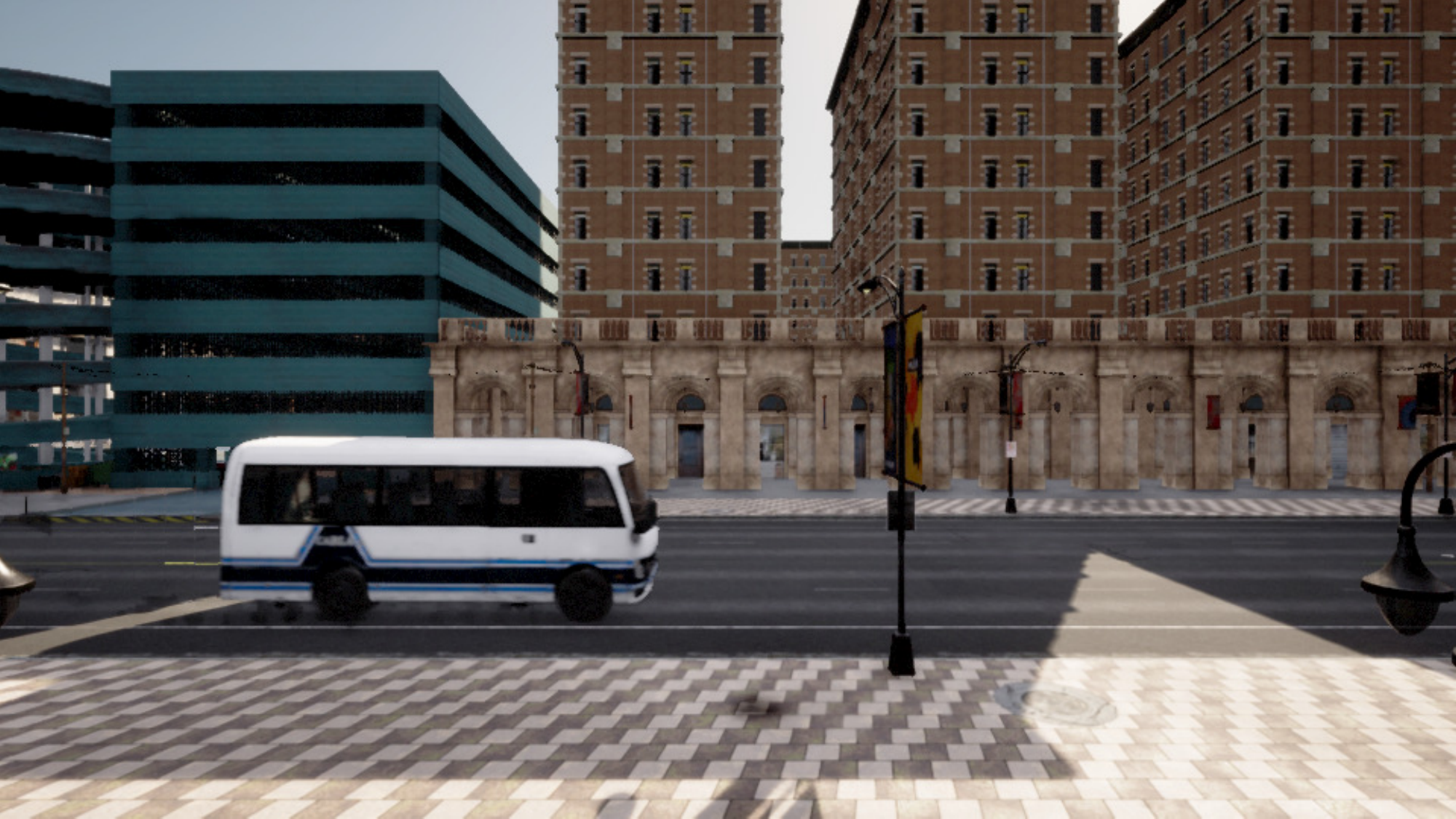}
		\caption{RGB frame captured by \ac{BS} 1's camera.}
		\label{fig: agent1_visual_frame}
	\end{subfigure}\hfill
    \begin{subfigure}[t]{0.45\textwidth}
		\centering   				%
		\includegraphics[width=\textwidth]{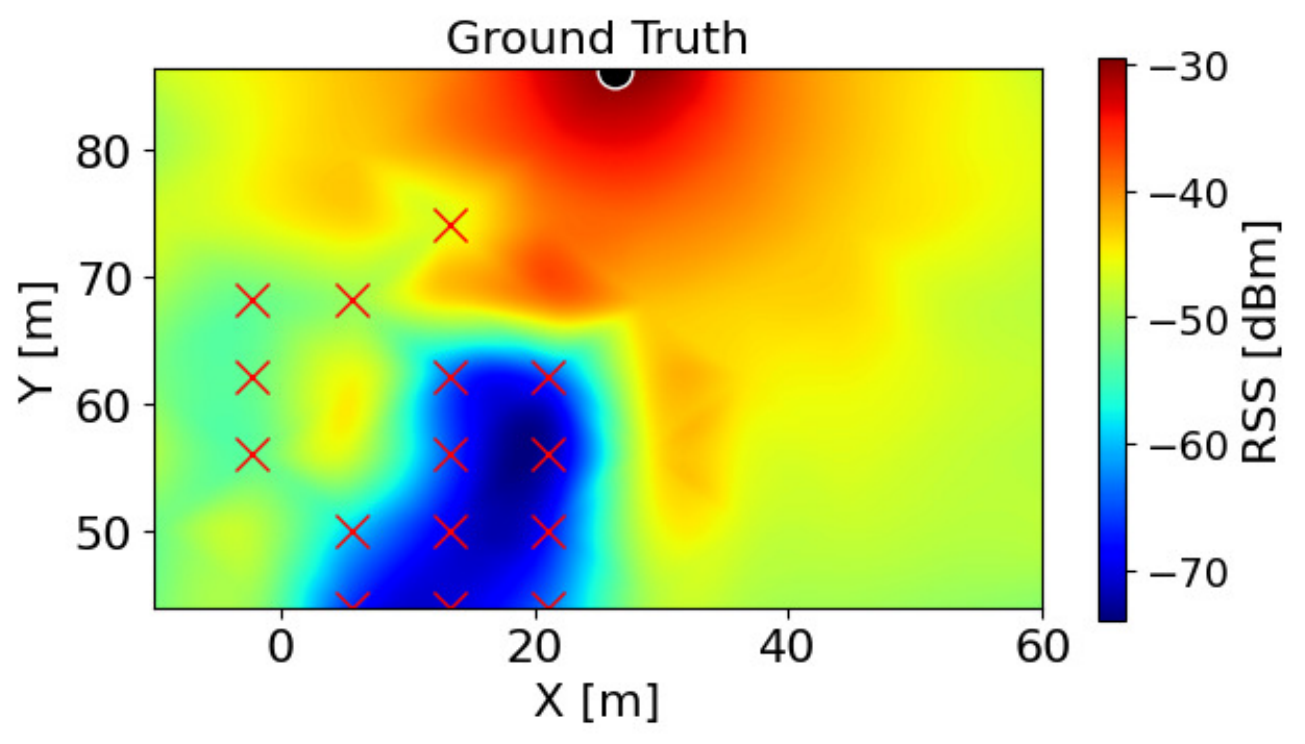}
		\caption{Corresponding true \ac{RSS} map of \ac{BS} 1.}
		\label{fig: agent1_visual_true}
	\end{subfigure}
	\vspace{1em}
	\begin{subfigure}[t]{0.45\textwidth}
        \centering
		\includegraphics[width=\textwidth]{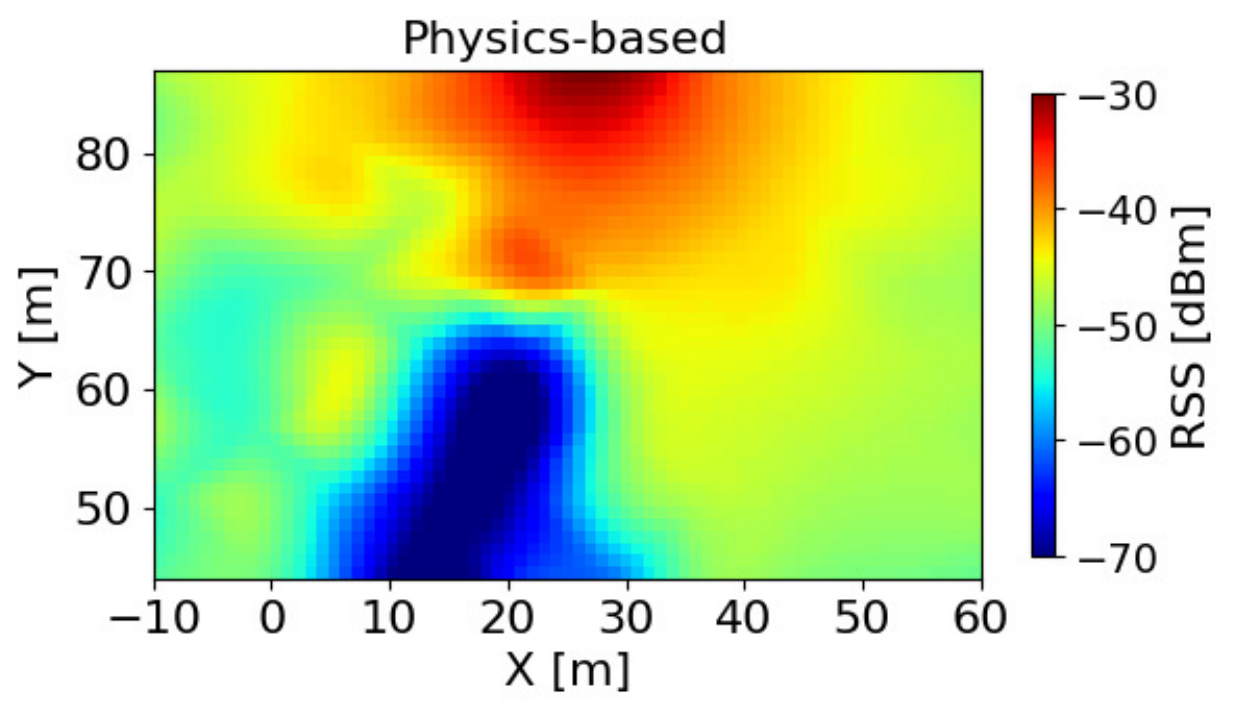}
		\caption{Predicted \ac{RSS} map of the physics-based model}
		\label{fig:agent1_visual_pred}
	\end{subfigure}\hfill
	\begin{subfigure}[t]{0.45\textwidth}
		\centering
		\includegraphics[width=\textwidth]{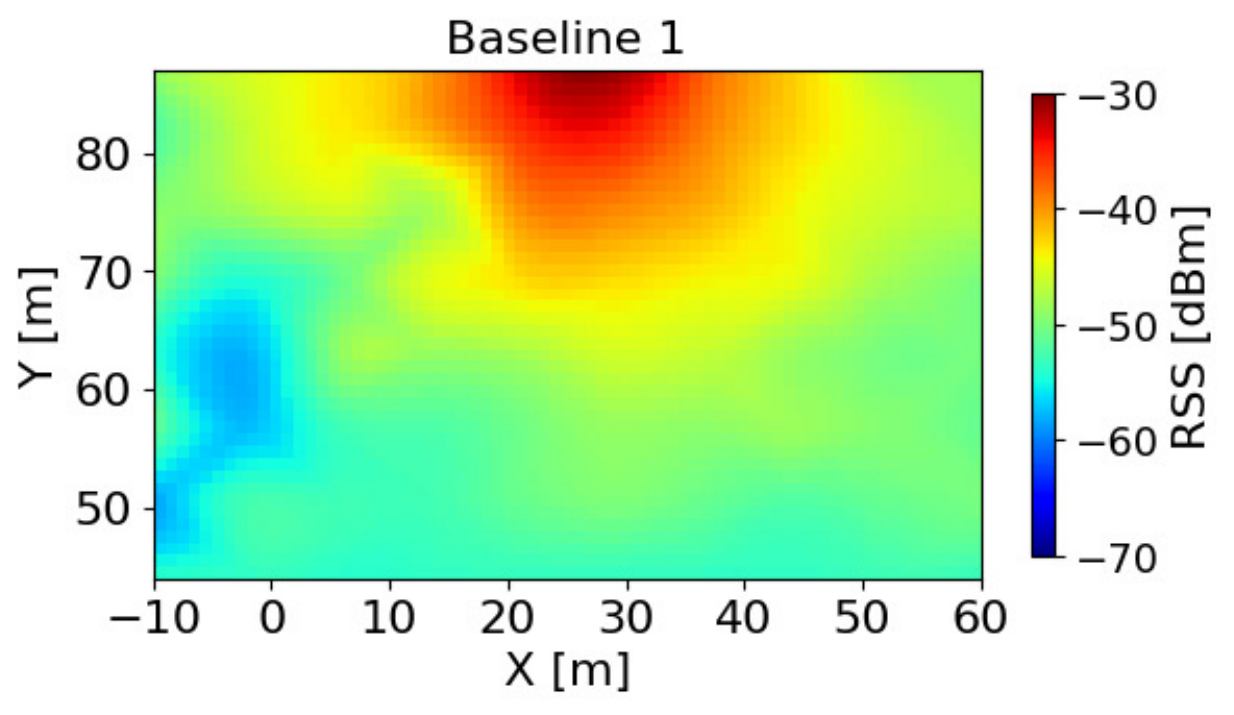}
	    \caption{Predicted \ac{RSS} map of the baseline model.}
	    \label{fig:agent1_visual_pred_baseline}
    \end{subfigure}
	    				\vspace{-0.3cm}
    \caption{Visualization of predicted \ac{RSS} maps of the physics-based and baseline models}
    \label{fig: agent1_visual}
\end{figure*}

Figure \ref{fig: agent1_visual} provides a visualization of the predicted \ac{RSS} maps of the proposed model and Baseline 1. Fig. \ref{fig: agent1_visual_frame} shows the RGB frame from the camera of \ac{BS} 1, and Fig. \ref{fig: agent1_visual_true} shows the ground truth \ac{RSS} map, where the red X represents receiver locations with only \ac{NLoS} paths. We can see that the bus creates a significant \ac{NLoS} area. This data sample is from the VAL-1 dataset. Hence, this frame is an unseen blockage pattern for both trained models. Figs. \ref{fig:agent1_visual_pred} and \ref{fig:agent1_visual_pred_baseline} show the predicted \ac{RSS} maps of the physics-based and Baseline 1 models, respectively. We can see that our model can accurately predict the \ac{NLoS} area caused by the bus even though this blockage pattern was not shown during the training. Meanwhile, Baseline 1 cannot predict the \ac{NLoS} region correctly. Although multi-modal data can provide richer information about the environment, vanilla training with \ac{MSE} loss still learns superficial features rather than learning invariant rules, thereby being sensitive to concept shifts.

\begin{figure*}[t!]
	\begin{subfigure}[t]{0.45\textwidth}
			\begin{center}   				%
				\includegraphics[width=\textwidth]{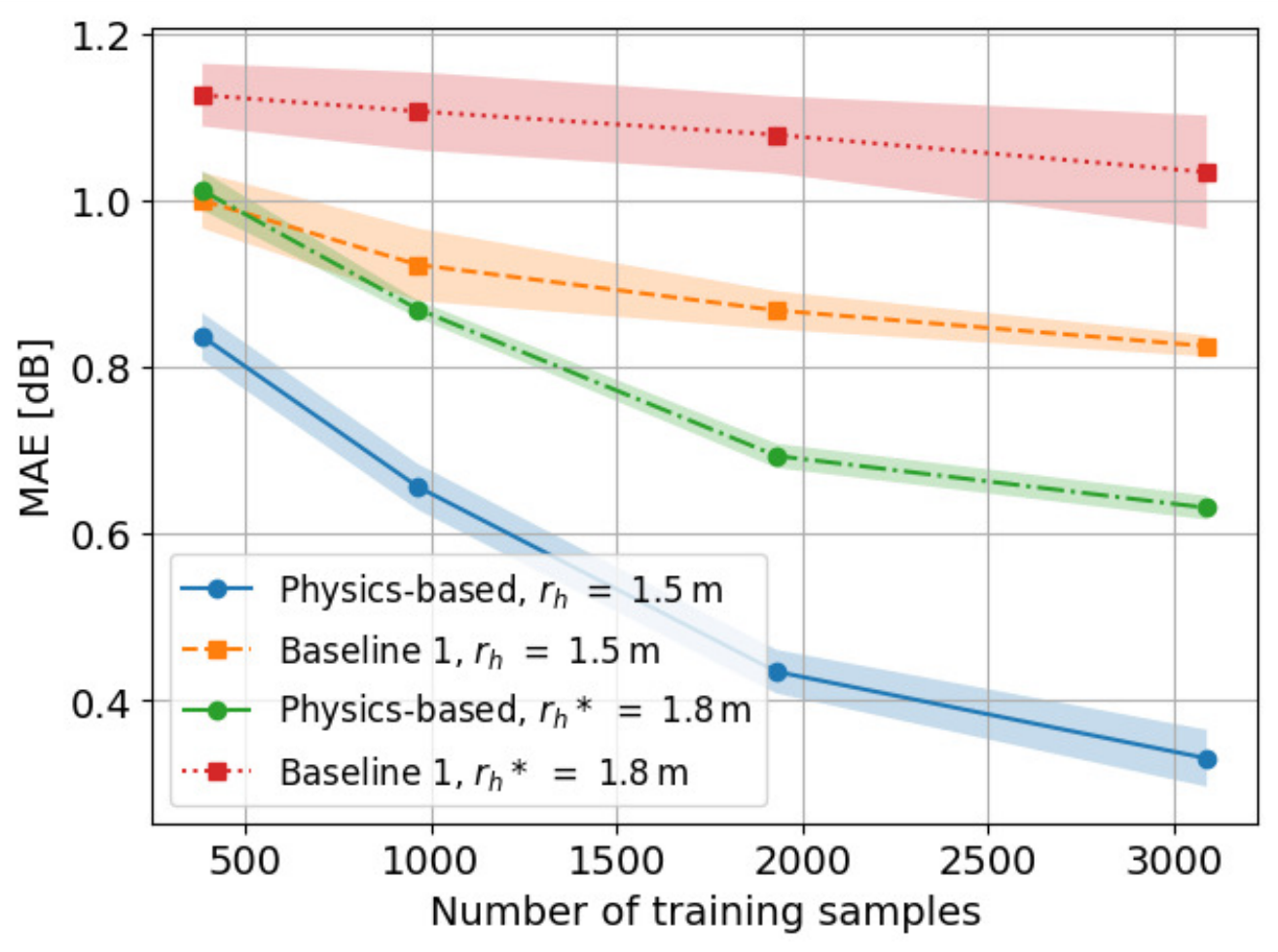}
		\end{center}
		    				\vspace{-0.3cm}
		\caption{\ac{MAE} with respect to the amount of training data of \ac{BS} 1 with unseen receiver height.}
		\label{fig: agent1_rx_1.8_MAE}
	\end{subfigure}\hfill
	\begin{subfigure}[t]{0.45\textwidth}
        \begin{center}   
		\includegraphics[width=\textwidth]{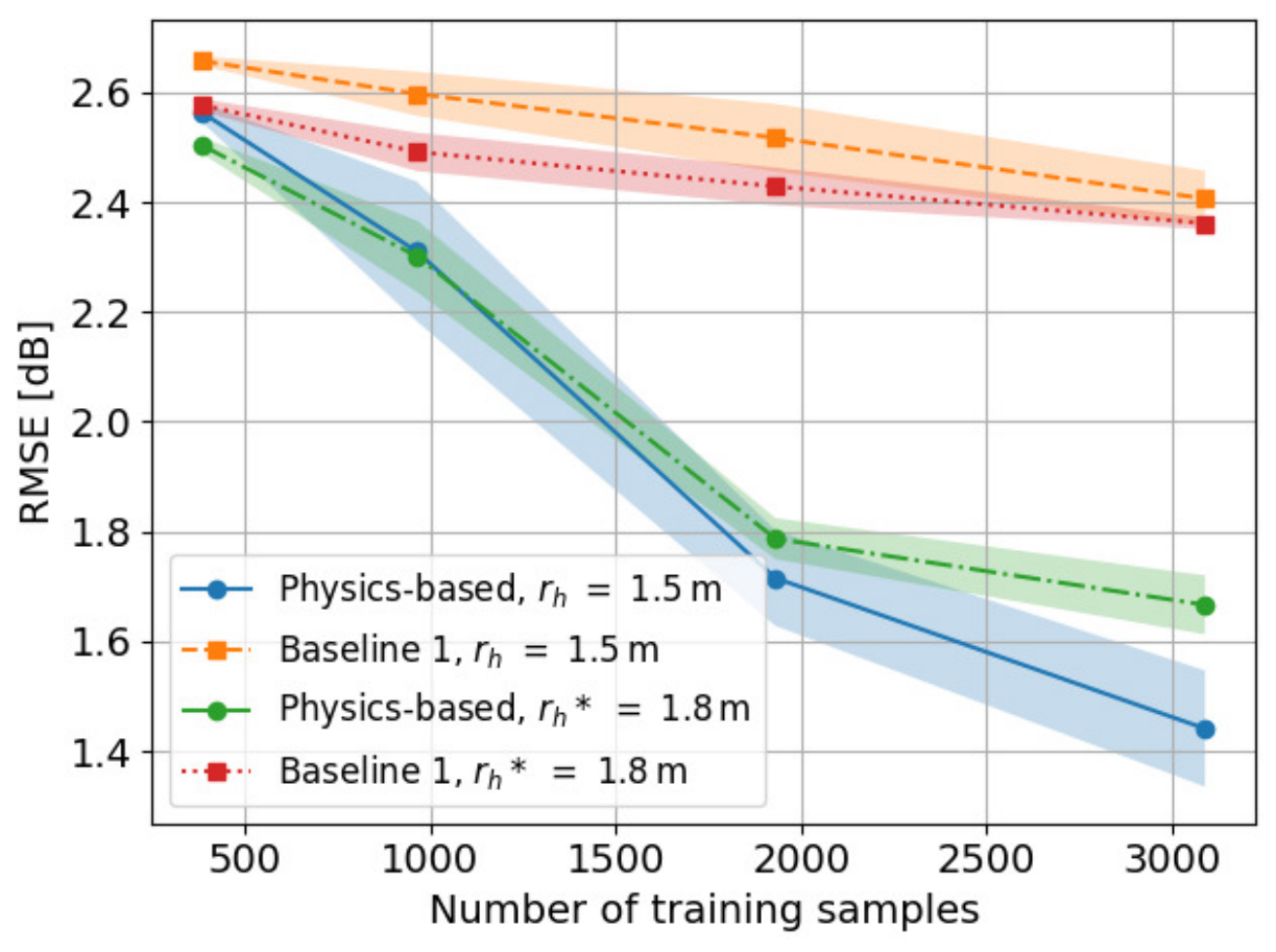}
		\end{center}
		    				\vspace{-0.3cm}
		\caption{\ac{RMSE} with respect to the amount of training data of \ac{BS} 1 with unseen receiver height.}
		\label{fig: agent1_rx_1.8_RMSE}
	\end{subfigure}\hfill
	    				\vspace{-0.2cm}
    \caption{Performance of \ac{BS} 1 with respect to the amount of training data with unseen receiver height.}
    \label{fig: agent1_rx_1.8}
    	     \vspace{-0.5cm}
\end{figure*}

Figure \ref{fig: agent1_rx_1.8} shows the performance of \ac{BS} 1 when receivers with unseen height appear. For the VAL-1 dataset, we set receiver height $\rxheight = 1.8$ m while we train models with $\rxheight = 1.5$ m. We mark the cases with the unseen receiver height as $\rxheight^* = 1.8$ m and the cases with the same receiver height as $\rxheight = 1.5$ m for the comparison. We can see that the performance of both the physics-based model and Baseline 1 degrades with the unseen receiver height. However, we can see that the \ac{RMSE} of the physics-based model is still close to the case when tested with $\rxheight = 1.5$ m. Hence, we can know that the proposed model still predicts \ac{NLoS} regions robustly even with unseen receiver heights.

\begin{figure*}[t!]
	\begin{subfigure}[t]{0.45\textwidth}
			\begin{center}   				%
				\includegraphics[width=\textwidth]{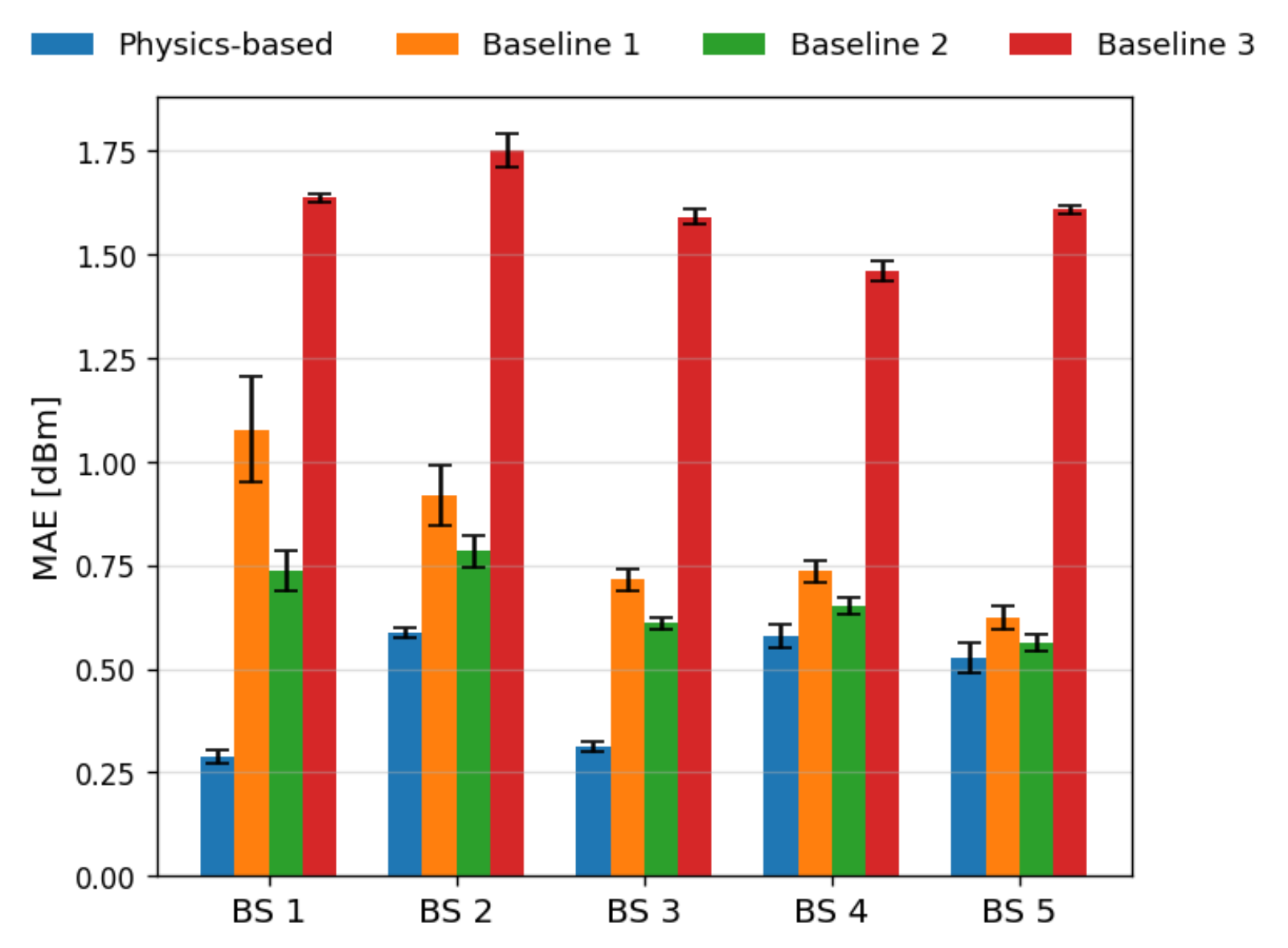}
		\end{center}
				\vspace{-0.3cm}
		\caption{\ac{MAE} of \acp{BS} on VAL-1 datasets.}
		\label{fig: agents_MAE_concept}
	\end{subfigure}\hfill
	\begin{subfigure}[t]{0.45\textwidth}
        \begin{center}   
		\includegraphics[width=\textwidth]{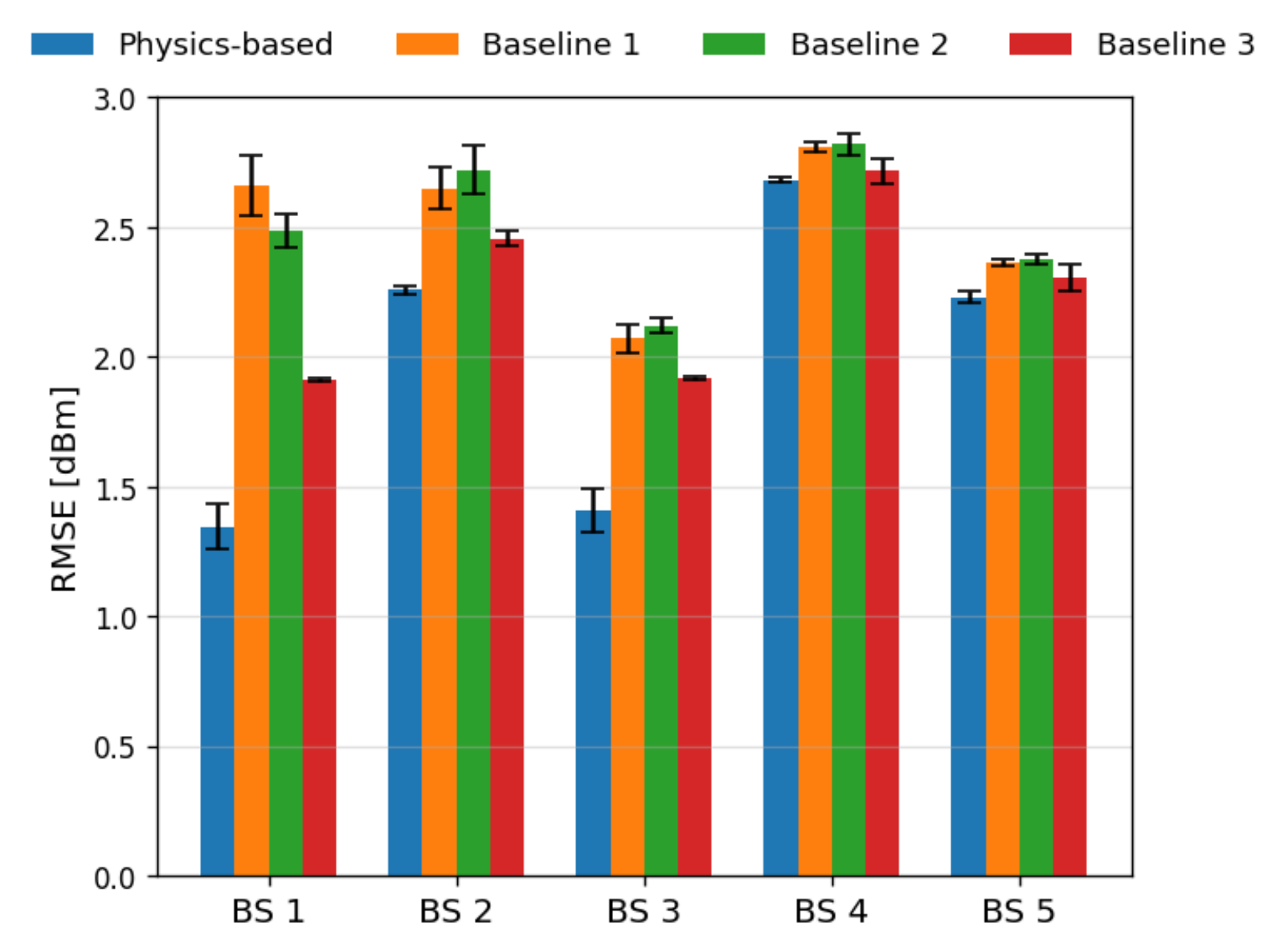}
		\end{center}
						\vspace{-0.3cm}
		\caption{\ac{RMSE} of \acp{BS} on VAL-1 datasets. }
		\label{fig: agents_RMSE_concept}
	\end{subfigure}\hfill
	    \vspace{-0.2cm}
    \caption{Performance of \acp{BS} on VAL-1 datasets.}
    \label{fig: agents_concept}
    	     \vspace{-0.3cm}
\end{figure*}

\begin{figure*}[t!]
	\begin{subfigure}[t]{0.45\textwidth}
		\begin{center}   				%
			\includegraphics[width=\textwidth]{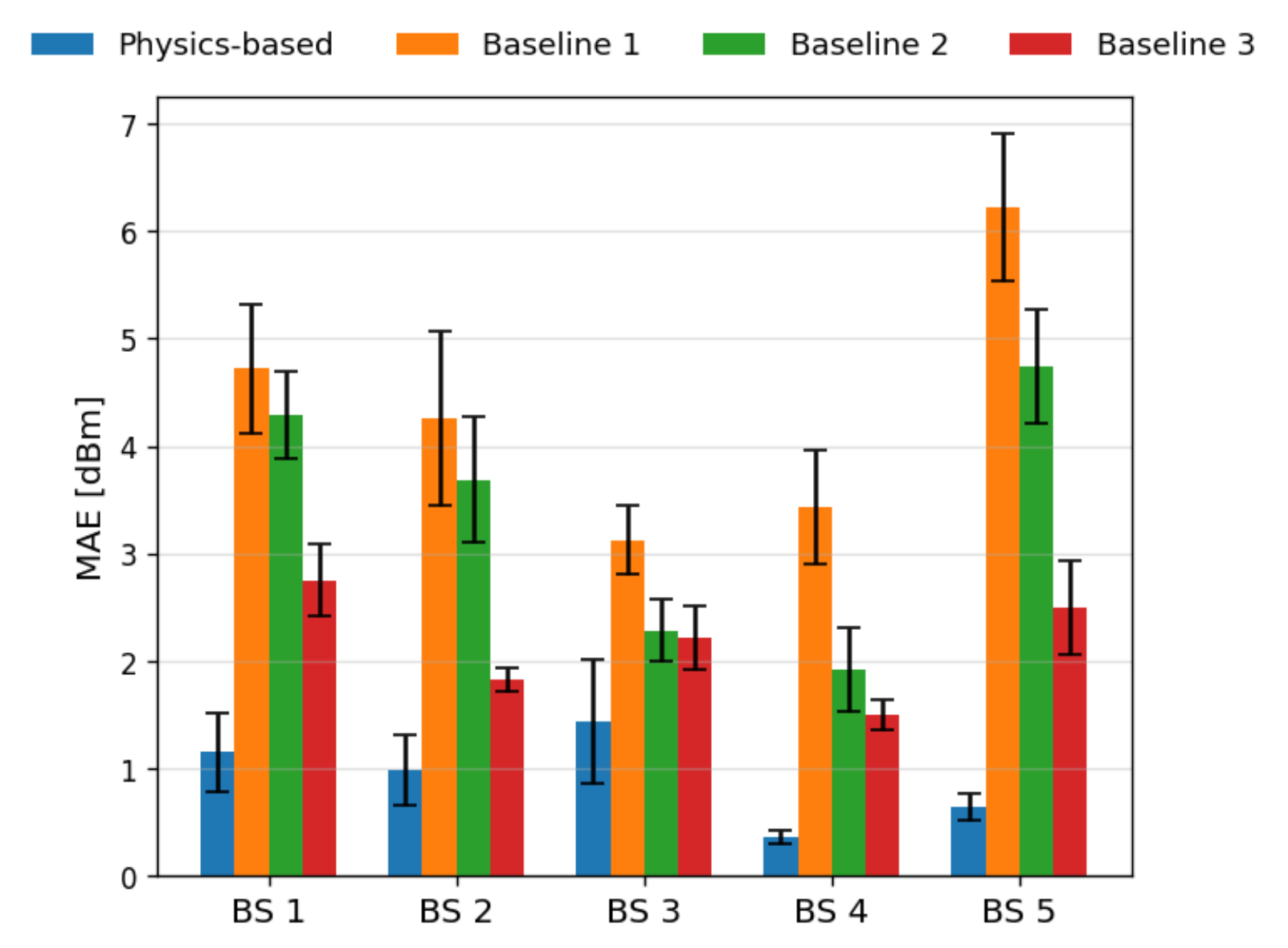}
		\end{center}
		\vspace{-0.3cm}
		\caption{\ac{MAE} of \acp{BS} on VAL-2 datasets.}
		\label{fig: agents_MAE_covariate}
	\end{subfigure}\hfill
	\begin{subfigure}[t]{0.45\textwidth}
		\begin{center}   
			\includegraphics[width=\textwidth]{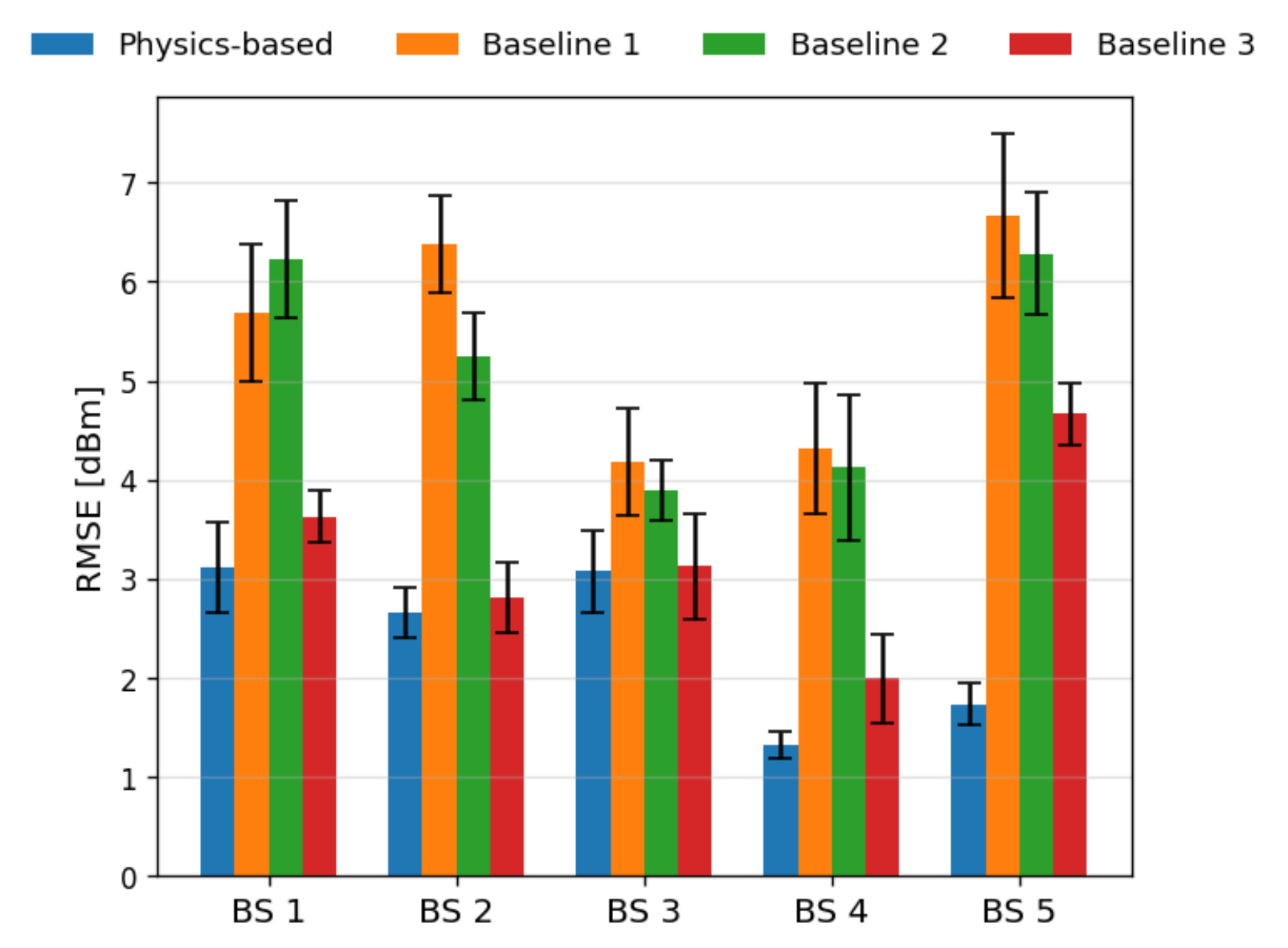}
		\end{center}
				\vspace{-0.3cm}
		\caption{\ac{RMSE} of \acp{BS} on VAL-2 datasets. }
		\label{fig: agents_RMSE_covariate}
	\end{subfigure}\hfill
			\vspace{-0.2cm}
	\caption{Performance of \acp{BS} on VAL-2 datasets.}
	\label{fig: agents_covariate}
		     \vspace{-0.3cm}
\end{figure*}

In Figs. \ref{fig: agents_concept}, we present the performance of our models and baselines on the VAL-1 datasets. The VAL-1 datasets induce both concept and covariate shifts. We can see that our proposed approach is more robust to the domain shifts than the baselines. In Fig. \ref{fig: agents_concept}, our models improve the \ac{MAE} and \ac{RMSE} by $44\%$ and $21\%$, respectively, compared to those of Baseline 1. For \acp{BS} 4 and 5, the performance is lower than that of other \acp{BS} because we induced stronger concept shifts by generating unseen large vehicles.  Baseline 2 can achieve better \ac{MAE} than Baseline 1 as it divides the \ac{MSE} loss function into \ac{LoS}/\ac{NLoS}, but its \ac{RMSE} is similar to that of Baseline 1. This is because Baseline 1 and 2 do not use any physical information for training. Our models improve the \ac{MAE} and \ac{RMSE} by 32\% and 21\%, respectively, compared to those of Baseline 2. Meanwhile,  Baseline 3 achieves higher \ac{MAE} on the VAL-1 than other baselines because it relies on only the 3GPP path loss for \ac{LoS} points, making it vulnerable to concept shifts. The \ac{MAE} and \ac{RMSE} of our models are smaller than those of Baseline 3 by 72\% and 12\%, respectively.

Fig. \ref{fig: agents_covariate} presents the performance of our models and baselines on the VAL-2 datasets, which induce only covariate shifts. We can see that our approach is significantly more robust to covariate shifts than the baselines. Baselines 1 and 2 are susceptible to covariate shifts because they focus on noise in multi-modal input rather than learning domain-invariant knowledge.  The \ac{MAE} and \ac{RMSE} of our models are smaller than those of Baseline 1 by $79\%$ and $56\%$, respectively. For Baseline 2, we can improve the \ac{MAE} and \ac{RMSE} by 73\% and 54\%, respectively. We can observe that Baseline 3 is more robust to covariate shifts than other baselines because it the 3GPP path loss for \ac{LoS} points. Hence, covariate shifts does not affect its prediction for \ac{LoS} points. However, Baseline 3 still shows the degradation for \ac{NLoS}. Our models can improve the \ac{MAE} and \ac{RMSE} by 58\% and 27\%. respectively, compared to Baseline 3. Hence, physics-based training can improve robustness to both covariate and concept shifts.

\begin{figure}[t]
		\centering 				%
		\begin{subfigure}[t]{0.8\linewidth}
			\centering
			\includegraphics[width=\textwidth]{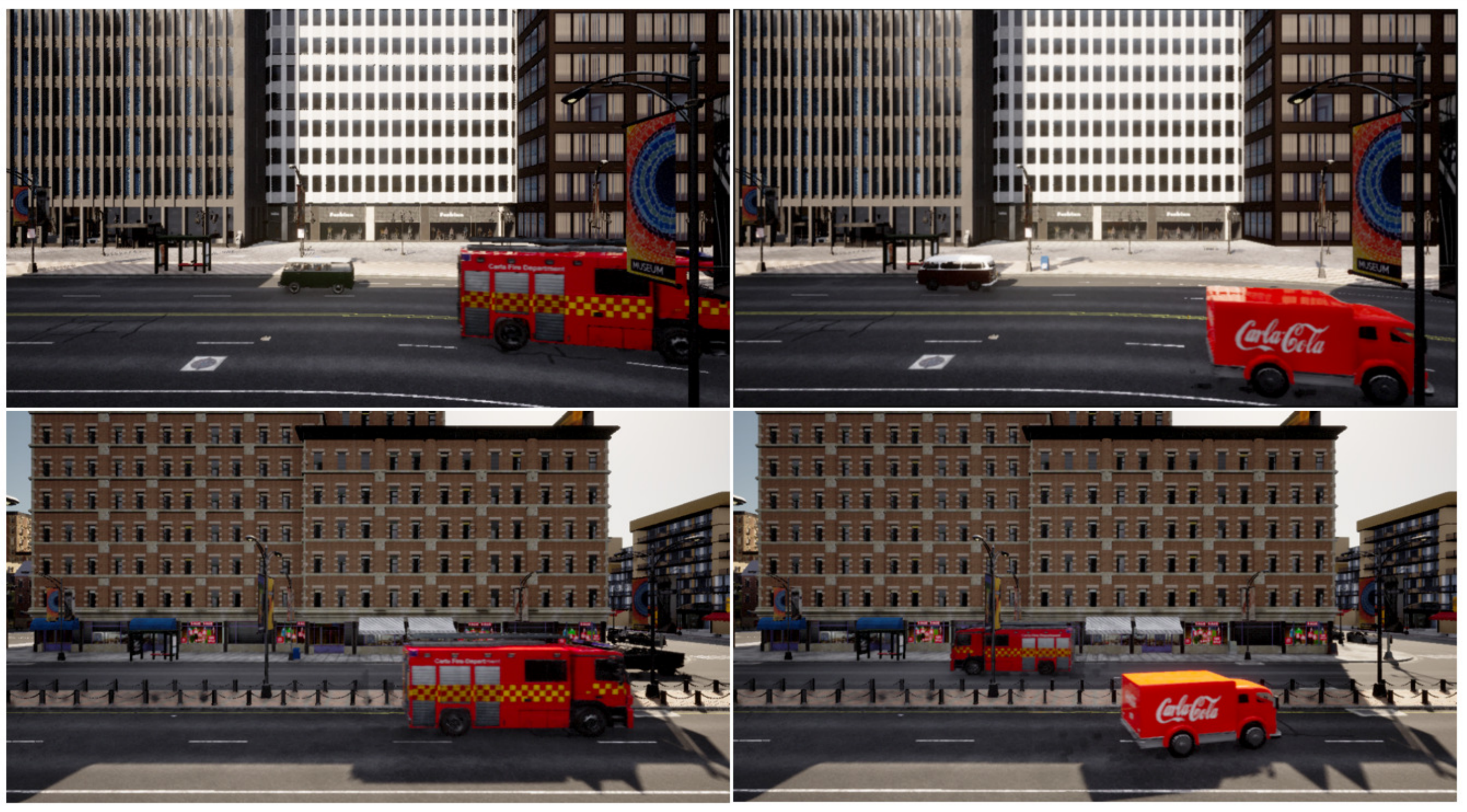}
		\end{subfigure}
		    \vspace{-0.2cm}
			\caption{First row: training data of \ac{BS} 3. Second row: VAL-1 dataset of \ac{BS} 5}
\label{fig: tta_example}
\vspace{-0.5cm}
\end{figure}

\begin{figure*}[t!]
	\begin{subfigure}[t]{0.45\textwidth}
			\begin{center}   				%
				\includegraphics[width=\textwidth]{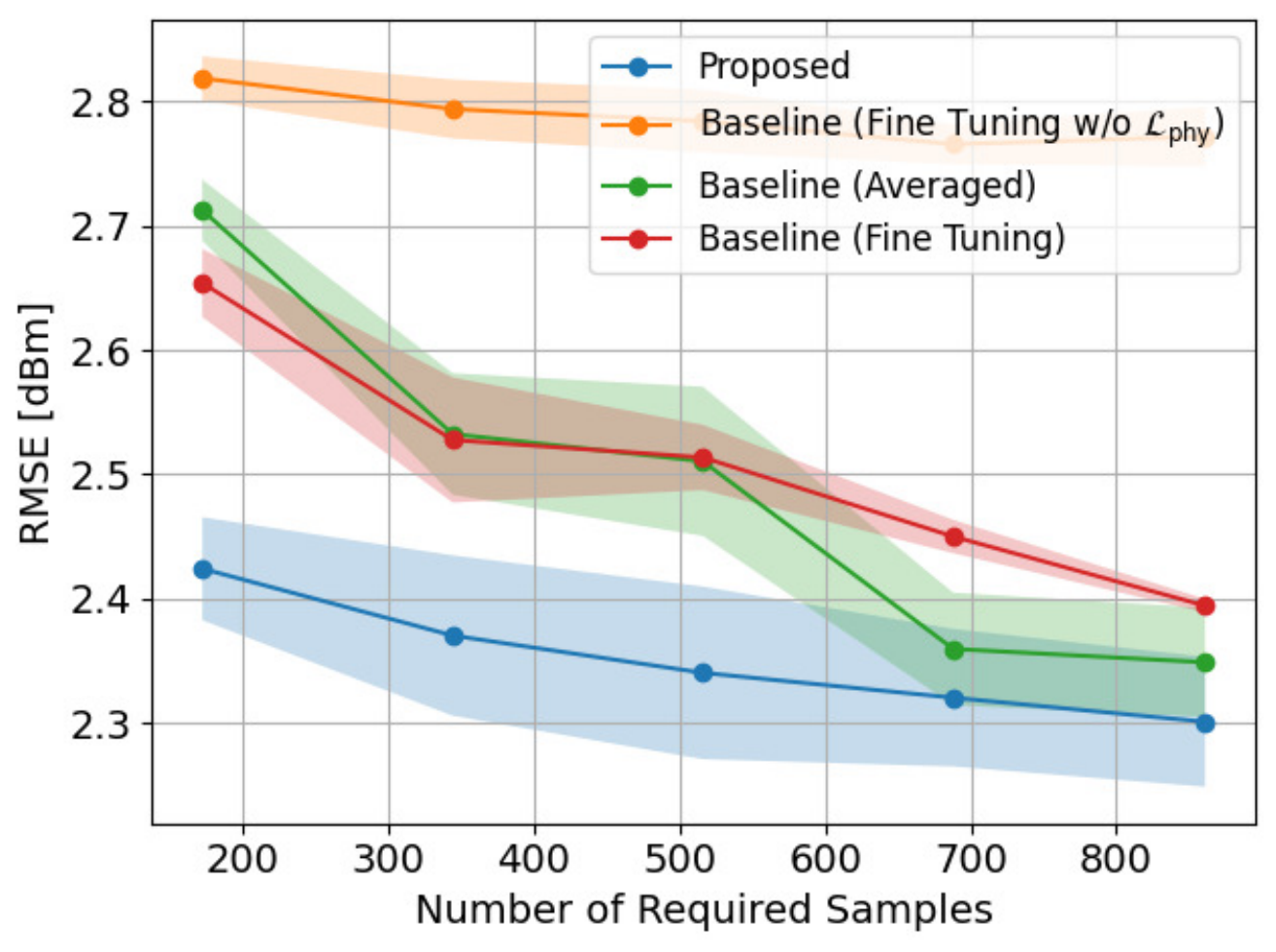}
		\end{center}
		    \vspace{-0.2cm}
		\caption{\ac{RMSE} of \ac{BS} 4 with respect to the amount of data.}
		\label{fig: agent4_tta_MAE}
	\end{subfigure}\hfill
	\begin{subfigure}[t]{0.45\textwidth}
        \begin{center}   
		\includegraphics[width=\textwidth]{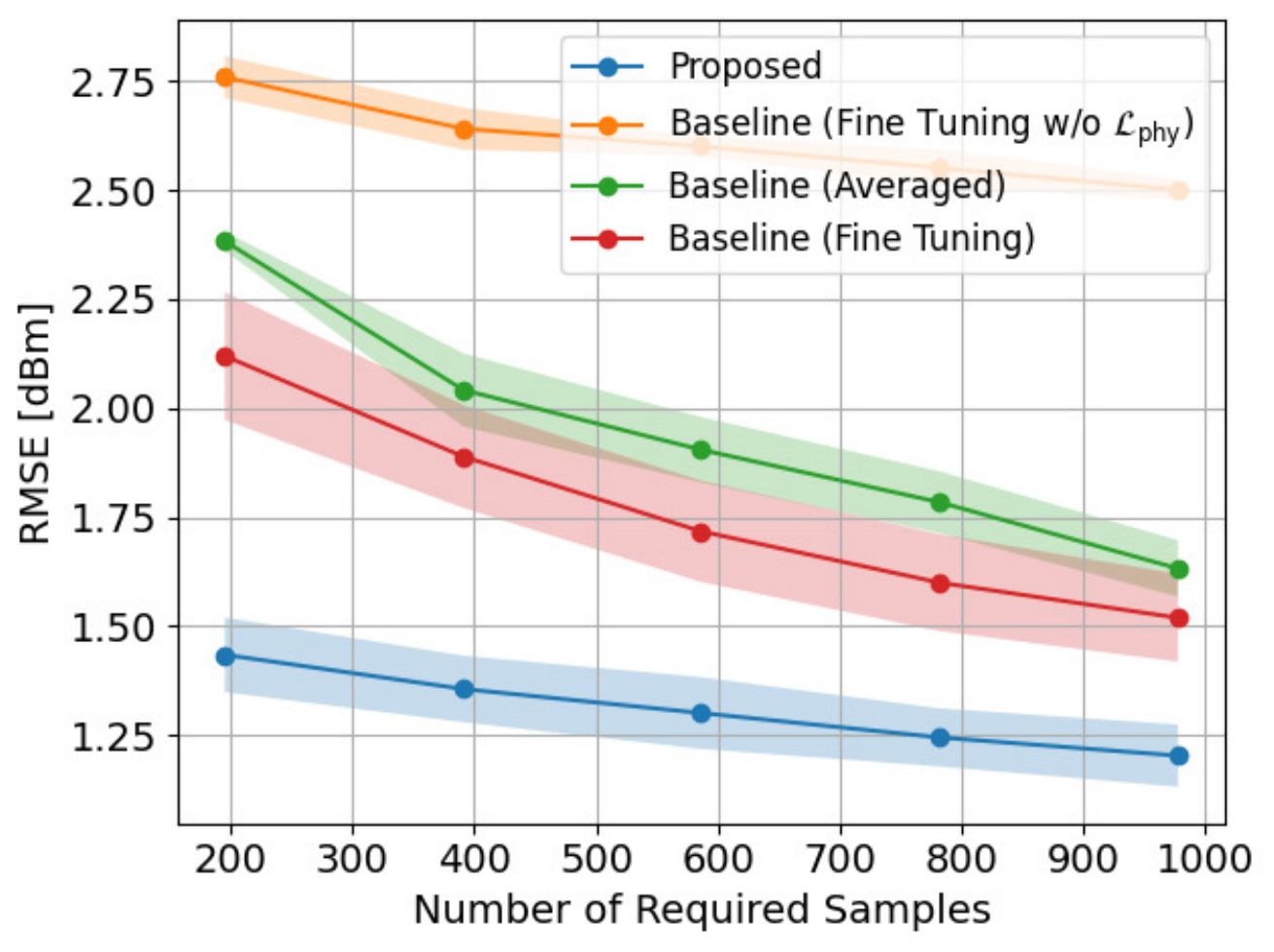}
		\end{center}
		    \vspace{-0.2cm}
		\caption{\ac{RMSE} of \ac{BS} 5 with respect to the amount of data. }
		\label{fig: agent5_tta_MAE}
	\end{subfigure}\hfill
	    \vspace{-0.2cm}
    \caption{Performance of \ac{BS} 4 and 5 under domain shifts with the proposed collaborative domain adaptation.}
    \label{fig: agent45_tta}
    	     \vspace{-0.3cm}
\end{figure*}

Figure \ref{fig: agent45_tta} presents the performance of the proposed collaborative domain adaptation approach. Here, \ac{BS}s 4 and 5 observe their VAL-1 datasets after the training. Hence, \ac{BS}s 4 and 5 need to adapt to the new environment to recover the performance loss as shown in Fig. \ref{fig: agents_MAE_concept}. We compare the proposed framework with three baselines. Baseline (Fine Tuning) locally fine-tunes the whole model. Baseline (Averaged) receives model parameters from other agents but, it performs aggregation in a way similar to federated learning's FedAvg algorithm \cite{10171192}. Baseline (Fine Tuning w.o, $\medloss_\text{phy}$) locally fine-tunes the entire model without using $\medloss_\text{phy}$.  To test the performance of the collaborative domain adaptation approach, we induced similar blockage patterns in the VAL-1 of \ac{BS}s 4 and 5 to the training data of \ac{BS}s 1, 2, and 3. A visualization of the example training data of \ac{BS} 3 and VAL-1 of \ac{BS} 5 is provided in Fig. \ref{fig: tta_example}. Hence, \ac{BS}s 1, 2, and 3  have useful knowledge that can be helpful to mitigate the domain shift of \ac{BS} 4 and 5. We can see that our framework outperforms the three baselines with respect to performance and data efficiency. We can observe that the proposed scheme requires significantly less data samples in test-time compared to the baselines. Our scheme needs 25\% of data to outperform the Baseline (Fine Tuning). \ac{BS} 4 and 5 can adapt to the current domain with minimal data samples because they use the knowledge of other agents to mitigate the domain shift. Baseline (Fine Tuning) performs better than other baselines due to the physics-based loss function. However, it needs more data than our scheme because it does not efficiently utilize the knowledge of other agents. Similarly, Baseline (Averaged) averages every model parameters that are not always helpful to mitigate the current domain shift, thereby decreasing the data efficiency.

\begin{figure*}[t!]
	\begin{subfigure}[t]{0.45\textwidth}
			\begin{center}   				%
				\includegraphics[width=\textwidth]{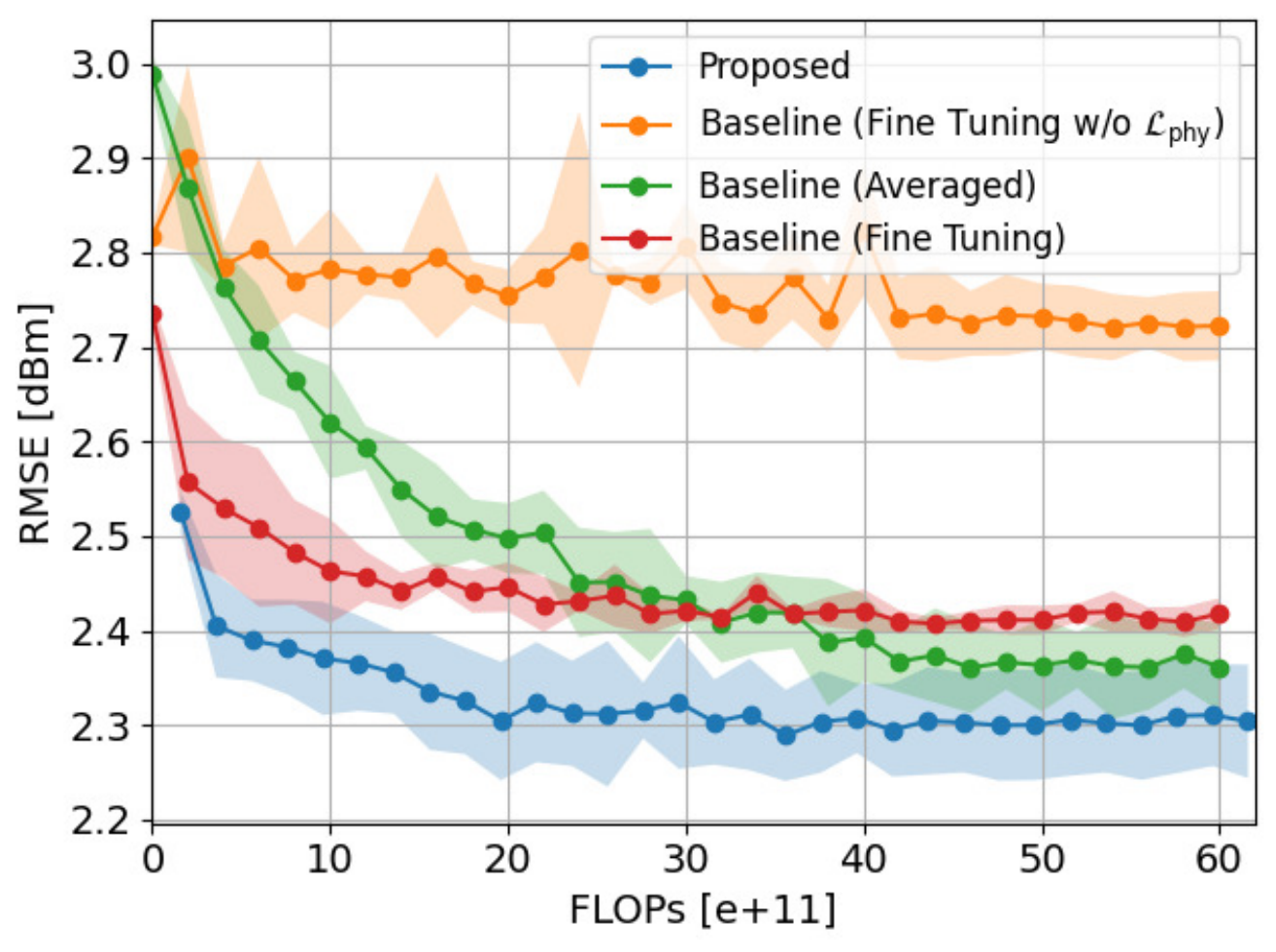}
		\end{center}
	    \vspace{-0.2cm}
		\caption{Learning curve of the \ac{RMSE} of \ac{BS} 4 under  domain shifts.}
		\label{fig: agent4_epoch}
	\end{subfigure}\hfill
	\begin{subfigure}[t]{0.45\textwidth}
        \begin{center}   
		\includegraphics[width=\textwidth]{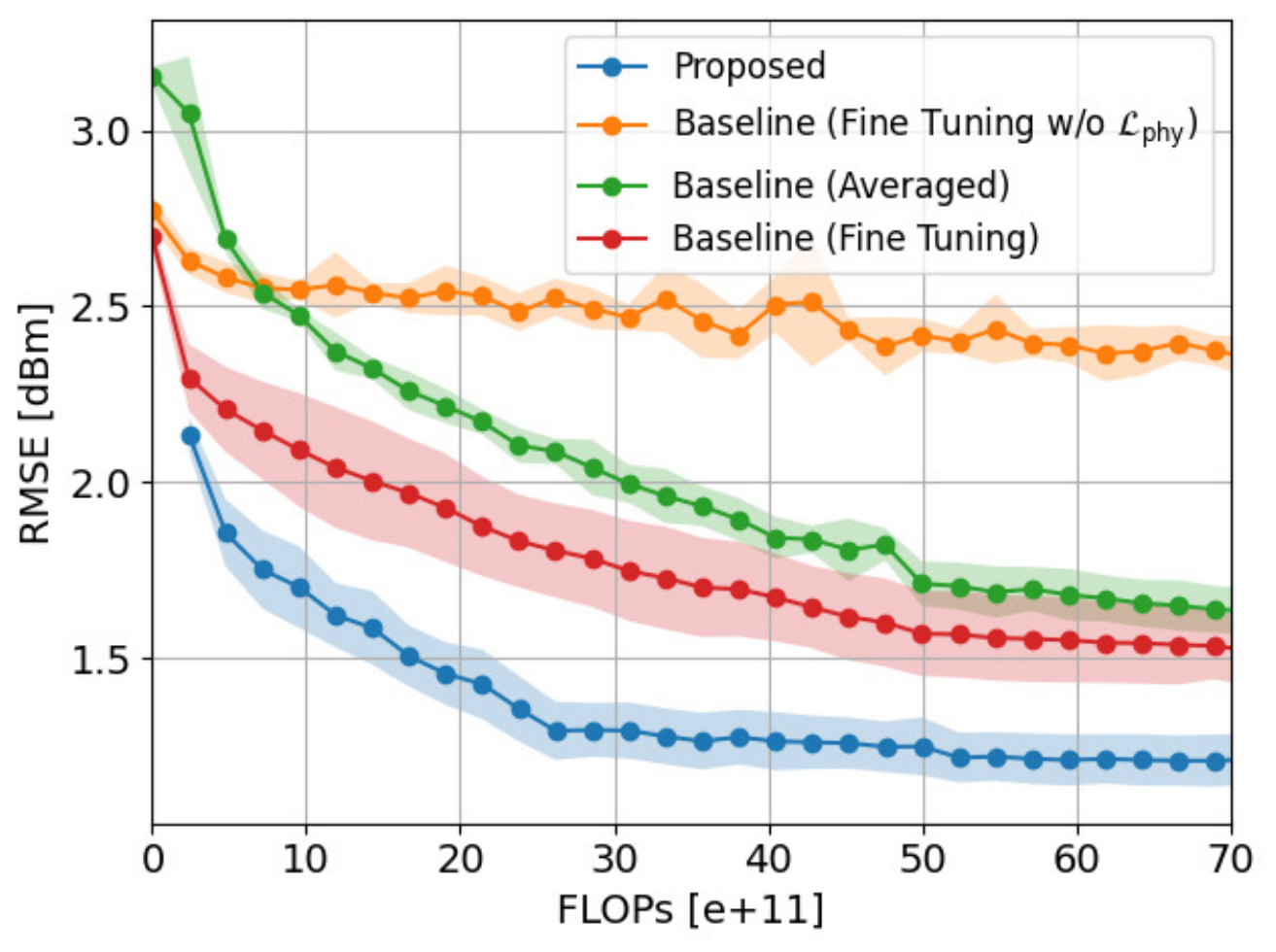}
		\end{center}
	    \vspace{-0.2cm}
		\caption{Learning curve of the \ac{RMSE} of \ac{BS} 5 under domain shifts. }
		\label{fig: agent5_epoch}
	\end{subfigure}\hfill
	    \vspace{-0.1cm}
    \caption{Learning curves of \ac{BS} 4 and 5 under domain shifts with the proposed collaborative domain adaptation.}
    \label{fig: agent45_epoch}
    	     \vspace{-0.3cm}
\end{figure*}

In Fig. \ref{fig: agent45_epoch}, we present the learning curves of \ac{BS} 4 and 5 under the domain shifts used in Fig. \ref{fig: agent45_tta} with respect to the number of \ac{FLOPs}. Here, we report \ac{FLOPs} to measure the compute overhead during the adaptation. We can see that the proposed collaborative domain adaptation converges significantly faster than the baselines with less computation. \Ac{BS} 4 and 5 can initiate their models using the domain similarity-based aggregation. Hence, the proposed aggregation gives a high weight to models that experienced similar domains. Note that our method first retrieves models from other \ac{BS}s, and does $\numagents-1$ inferences to calculate the domain similarity as shown in Algorithm. \ref{algo:1}, leading to additional computation and communication overhead. The communication overhead negligible as the model size is 50 Mb and the backhaul often provides throughput of Gbps scale \cite{tezergil2022wireless}. Calculating domain similarities requires additional forward passes. Hence, in Fig. \ref{fig: agent45_epoch}, we can see that the curves of \ac{BS}s 4 and 5 start later than other baselines that do fine-tuning immediately. However, the baselines show slower convergence due to the lack of enough available data for the adaptation. For Baseline (Averaged), the aggregated model is not helpful to mitigate the current domain shift, and it degrades the performance due to the aggregation of irrelevant models.

\begin{figure*}[t!]
	\begin{subfigure}[t]{0.45\textwidth}
		\begin{center}   				%
			\includegraphics[width=\textwidth]{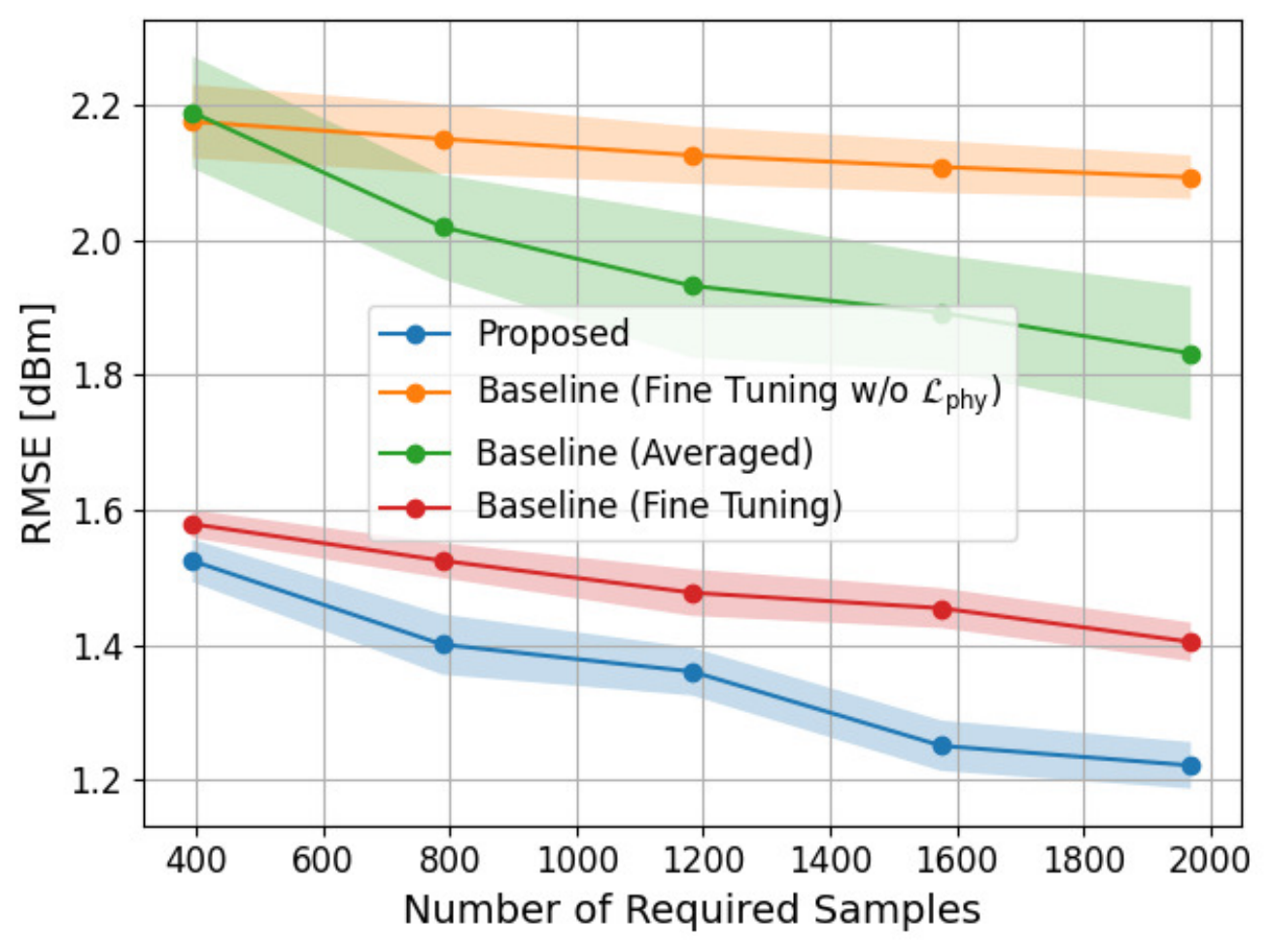}
		\end{center}
	     \vspace{-0.3cm}
		\caption{Learning curve of the \ac{RMSE} of \ac{BS} 2 under  domain shifts.}
		\label{fig: agent2_tta_RMSE}
	\end{subfigure}\hfill
	\begin{subfigure}[t]{0.45\textwidth}
		\begin{center}   
			\includegraphics[width=\textwidth]{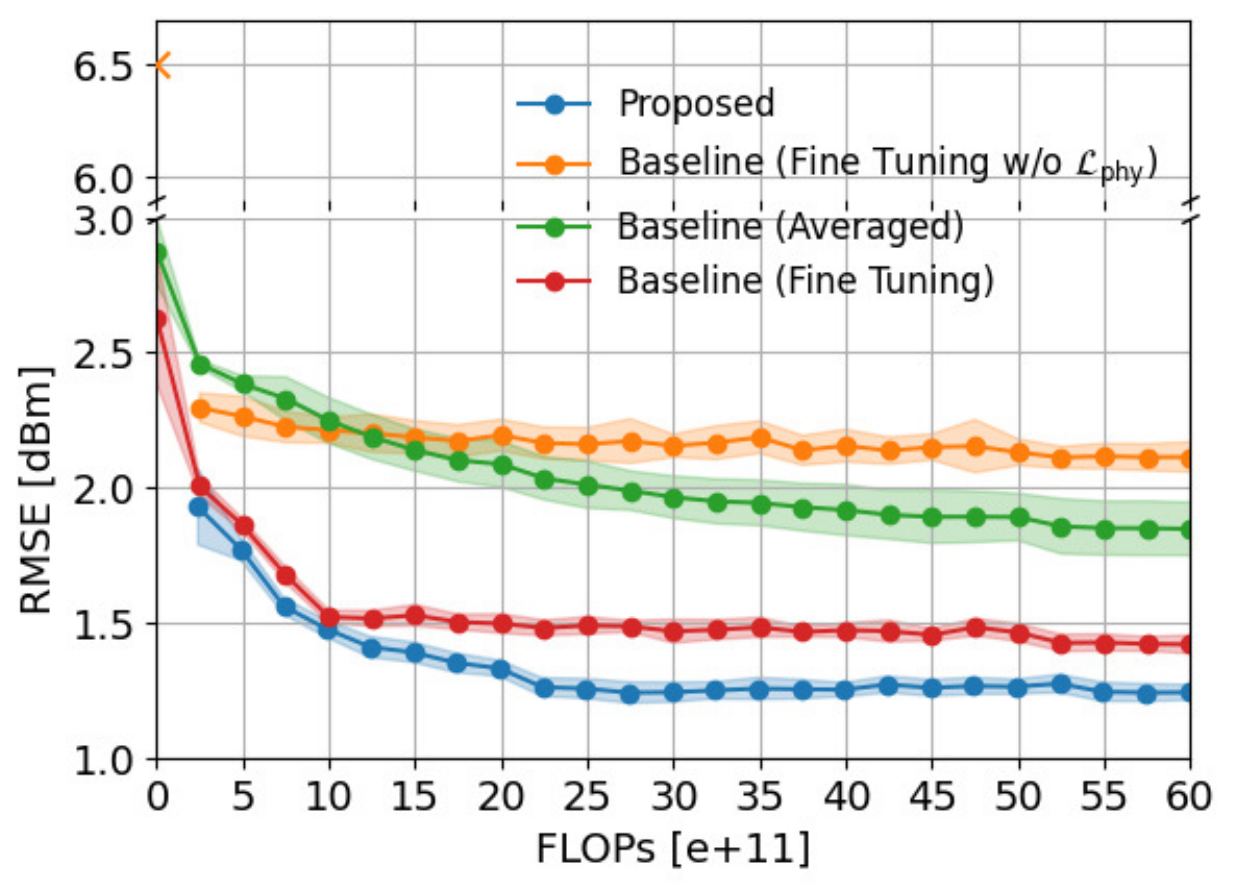}
		\end{center}
		     \vspace{-0.3cm}
		\caption{Learning curve of the \ac{RMSE} of \ac{BS} 2 under domain shifts. }
		\label{fig: agent2_epoch}
	\end{subfigure}\hfill
	\caption{Learning curves of \ac{BS} 2 under domain shifts with the proposed collaborative domain adaptation.}
	\label{fig: agent2_tta}
		     \vspace{-0.3cm}
\end{figure*}

Figure \ref{fig: agent2_tta} presents the performance of the proposed collaborative domain adaptation with \ac{BS} 2 using its VAL-2 dataset. To simulate the domain shift, \ac{BS}2 observes its VAL-2 dataset after the training. To test the performance of the collaborative domain adaptation, we induced the same covariate shift patterns of the VAL-2 to the training datasets of \ac{BS}s 4 and 5. Hence, \ac{BS}s 4 and 5 have useful knowledge that guides \ac{BS} 2 under the domain shift. We can see that our framework outperforms the three baselines with respect to the data efficiency and adaptation speed. In Fig. \ref{fig: agent2_tta_RMSE}, we can see that our scheme needs only 40\% of data to outperform the Baseline (Fine Tuning). \Ac{BS} 2 shows high data efficiency because it utilizes the knowledge of \acp{BS} 4 and 5 to mitigate the current domain shift. Fig. \ref{fig: agent2_tta} also shows that our method converges faster to better \ac{RMSE} than other baselines because it starts adaptation by using the knowledge of \acp{BS} 4 and 5.

\section{Conclusion} \label{sec: med_conclusion}

In this paper, we have studied the problems of achieving robust generalization to dynamic wireless environments with limited multi-modal data. We have proposed to use physics-based training to improve the data efficiency and robustness to domain shifts. We have presented the theoretical insights on the impact of physics-based training on data efficiency. Then, we have presented the collaborative domain adaptation framework to utilize the knowledge of other \ac{BS}s to help under-performing \ac{BS}s under domain shifts to improve the data efficiency and adaptation speed. To this end, we have proposed to measure the domain similarity between each \ac{BS} to effectively aggregate model parameters of \ac{BS}s. We have provided one use case of our framework for predicting \ac{RSS} maps of multiple \acp{BS} based on multi-modal data. To validate our framework, we have developed a novel multi-modal dataset that can simulate the proposed use case and induce domain shifts. Simulation results validate the data efficiency, generalization performance, and our theoretical analysis of the proposed frameworks. The proposed algorithms require much less training data than a standard fine-tuning and the baselines to achieve convergence. Moreover, our framework shows stronger generalization to domain shifts under dynamic wireless environments. In essence, this work provides the first systematic study on how to design a learning framework to achieve robust generalization with limited multi-modal data under dynamic wireless environments.

\appendix
\subsection{Proof Theorem \ref{theorem:1}} \label{proof: theorem1}
The proof is inspired by \cite{garg2020functional}. First, we derive that $\nn \in \hypothesis(\epsilon_0)$ achieves $\medloss_\text{phy} (\medout, \nn(\medmultimodalin)) = 0$ with a certain probability. We assume a set of $\nn \notin \hypothesis(\epsilon_0)$, where $\hypothesis(\epsilon_0) = \{ \nn \in \hypothesis : \bar{\medloss}_\text{phy}(\medout, \theta(\medmultimodalin)) \leq \epsilon_0 \}$. For each data sample $\medmultimodalin_i$, we define an indicator variable $z_{\text{phy}, i} = \mathbbm{1} \left[ l_\text{phy}(\medout, \nn(\medmultimodalin_i) = 0  \right]$. Assume that $z_{\text{phy}, i}$ follows a Bernoulli distribution such that $p = \P[z_{\text{phy}, i} = 1], \forall i$. Then, $\bar{\medloss}_\text{phy}(\medout, \nn(\medmultimodalin)$ can be bounded as follows
\begin{align}
	\bar{\medloss}_\text{phy}(\nn(\medmultimodalin)) &= \E_{(\medmultimodalin, \medout) \sim \P(\medmultimodalin, \medout)}[l_\text{phy} (\medout, \theta(\medmultimodalin)) ]  \\
		&= \E_{(\medmultimodalin, \medout) }[l_\text{phy} (\medout, \theta(\medmultimodalin)) ] \P[l_\text{phy} (\medout, \theta(\medmultimodalin)) = 0] 
		\ka
		&\quad +  \E_{(\medmultimodalin, \medout)}[l_\text{phy}(\medout, \theta(\medmultimodalin)) ] \P[l_\text{phy} (\medout, \theta(\medmultimodalin)) > 0] \ka
		&= (1-p) \E_{(\medmultimodalin, \medout) \sim \P(\medmultimodalin, \medout)}[l_\text{phy} (\medout, \theta(\medmultimodalin)) ] \\
		&< (1-p),
	\end{align}
	where the last inequality comes from the assumption that $l_\text{phy} (\medout, \theta(\medmultimodalin)) \in [0, 1].$ Since $\nn \notin \hypothesis(\epsilon_0)$, we can bound $p$ as follows
	\begin{align}
		p < 1 - \bar{\medloss}_\text{phy}(\medout, \nn(\medmultimodalin))  < 1 - \epsilon_0.
	\end{align}
	Then, for $\nn \notin \hypothesis(\epsilon_0)$ and the number of data samples $m$, the probability of the empirical loss $\medloss_\text{phy} (\medout, \nn(\medmultimodalin))$ to be zero can be bounded as follows
	\begin{align}
		\P[\medloss_\text{phy} (\medout, \nn(\medmultimodalin)) = 0] < (1-\epsilon_0)^m.
	\end{align}
	Then, with probability at least $1 - (1 - \epsilon_0)^m$, only $\nn$ with $\medloss_\text{phy} (\medout, \nn(\medmultimodalin)) \leq \epsilon_0$ have $\medloss_\text{phy} (\medout, \nn(\medmultimodalin)) = 0$. Hence, only $\nn \in \hypothesis(\epsilon_0)$ achieve $\medloss_\text{phy} (\medout, \nn(\medmultimodalin)) = 0$. 
	
	Next, we show that only $\nn$ with $\bar{\medloss}_\text{data}(\medout, \nn(\medmultimodalin)) \leq \epsilon_1$ can achieve $\bar{\medloss}_\text{phy}(\medout, \nn(\medmultimodalin)) \leq \epsilon_0$ with a certain probability. We follow similar steps as above. Consider $\nn \in \hypothesis(\epsilon_0)$ and $\bar{\medloss}_\text{data} (\medout, \nn(\medmultimodalin)) > \epsilon_1$. We define $z_{\text{data}, i} = \mathbbm{1} \left[  
	l_\text{data}(\medout_i, \nn(\medmultimodalin_i)) = 0
	\right]$ 
	and assume it follows a Bernoulli distribution such that $p' = \P[z_{\text{data}, i} = 1], \forall i.$ Then, $\bar{\medloss}_\text{data}(\medout, \nn(\medmultimodalin))$ can be bounded as $\bar{\medloss}_\text{data}(\medout, \nn(\medmultimodalin)) < 1-p'$ as done above. Similarly, we bound $p'$ as follows
	\begin{align}
		p' < 1 - \bar{\medloss}_\text{data}(\medout, \nn(\medmultimodalin)) < 1 - \epsilon_1.
	\end{align}
	Then, for $\nn \in \hypothesis(\epsilon_0)$ and $\bar{\medloss}_\text{data}(\medout, \nn(\medmultimodalin)) > \epsilon_1$, the probability of $\medloss_\text{data}(\medout, \nn(\medmultimodalin))$ to be zero can be bounded as
	\begin{align}
		\P[\medloss_\text{data}(\medout, \nn(\medmultimodalin)) = 0] &< (1-\epsilon_1)^m \\
		&< |\hypothesis(\epsilon_0)| (1-\epsilon_1)^m \label{union_bound} \\
		&< |\hypothesis(\epsilon_0)| \exp(-\epsilon_1 m) \label{exp},
	\end{align}
	where \eqref{union_bound} is from the union bound, and the last inequality is from $(1-x)^y \leq \exp(-xy)$. By bounding \eqref{exp} with a certain probability $\delta \in (0, 1)$, we can derive the bound of $m$ to satisfy \eqref{exp} as follows
	\begin{align}
		m \geq \frac{1}{\epsilon_1} \left[  
		\ln|\hypothesis(\epsilon_0)| + \ln\frac{1}{\delta}
		\right].
	\end{align}
	Therefore, with probability at least $1 - \delta$, only $\nn \in \hypothesis(\epsilon_0)$ and $\hypothesis(\epsilon_1)$ can have $\P[\medloss_\text{data}(\medout, \nn(\medmultimodalin)) = 0]$ with the given value of $m$, thereby proving the theorem.  
	
\bibliographystyle{IEEEtran}
\bibliography{Bibtex/StringDefinitions,Bibtex/IEEEabrv,Bibtex/mybib}

\begin{thebibliography}{10}
\providecommand{\url}[1]{#1}
\csname url@samestyle\endcsname
\providecommand{\newblock}{\relax}
\providecommand{\bibinfo}[2]{#2}
\providecommand{\BIBentrySTDinterwordspacing}{\spaceskip=0pt\relax}
\providecommand{\BIBentryALTinterwordstretchfactor}{4}
\providecommand{\BIBentryALTinterwordspacing}{\spaceskip=\fontdimen2\font plus
\BIBentryALTinterwordstretchfactor\fontdimen3\font minus
  \fontdimen4\font\relax}
\providecommand{\BIBforeignlanguage}[2]{{%
\expandafter\ifx\csname l@#1\endcsname\relax
\typeout{** WARNING: IEEEtran.bst: No hyphenation pattern has been}%
\typeout{** loaded for the language `#1'. Using the pattern for}%
\typeout{** the default language instead.}%
\else
\language=\csname l@#1\endcsname
\fi
#2}}
\providecommand{\BIBdecl}{\relax}
\BIBdecl

\bibitem{10929033}
W.~Saad, O.~Hashash, C.~K. Thomas, C.~Chaccour, M.~Debbah, N.~Mandayam, and
  Z.~Han, ``Artificial general intelligence ({AGI})-native wireless systems: A
  journey beyond {6G},'' \emph{Proceedings of the IEEE}, pp. 1--39, 2025.

\bibitem{park2025resource}
Y.~M. Park, Y.~K. Tun, W.~Saad, and C.~S. Hong, ``Resource-efficient beam
  prediction in mmwave communications with multimodal realistic simulation
  framework,'' \emph{arXiv preprint arXiv:2504.05187}, 2025.

\bibitem{10949588}
L.~Bai, Z.~Huang, M.~Sun, X.~Cheng, and L.~Cui, ``Multi-modal intelligent
  channel modeling: A new modeling paradigm via synesthesia of machines,''
  \emph{IEEE Communications Surveys and Tutorials}, Apr. 2025.

\bibitem{alkhateeb2023deepsense}
A.~Alkhateeb, G.~Charan, T.~Osman, A.~Hredzak, J.~Morais, U.~Demirhan, and
  N.~Srinivas, ``Deepsense {6G}: A large-scale real-world multi-modal sensing
  and communication dataset,'' \emph{IEEE Communications Magazine}, vol.~61,
  no.~9, pp. 122--128, Sep. 2023.

\bibitem{schuhmann2022laion}
C.~Schuhmann, R.~Beaumont, R.~Vencu, C.~Gordon, R.~Wightman, M.~Cherti,
  T.~Coombes, A.~Katta, C.~Mullis, M.~Wortsman \emph{et~al.}, ``Laion-5b: An
  open large-scale dataset for training next generation image-text models,''
  \emph{Advances in neural information processing systems}, vol.~35, pp.
  25\,278--25\,294, 2022.

\bibitem{qu2023modality}
S.~Qu, Y.~Pan, G.~Chen, T.~Yao, C.~Jiang, and T.~Mei, ``Modality-agnostic
  debiasing for single domain generalization,'' in \emph{Proceedings of the
  IEEE/CVF Conference on Computer Vision and Pattern Recognition}, 2023, pp.
  24\,142--24\,151.

\bibitem{zhang2025out}
X.~Zhang, J.~Li, W.~Chu, R.~Xu, Y.~Yang, S.~Guan, J.~Xu, L.~Jing, P.~Cui
  \emph{et~al.}, ``On the out-of-distribution generalization of large
  multimodal models,'' in \emph{Proceedings of the Computer Vision and Pattern
  Recognition Conference}, 2025, pp. 10\,315--10\,326.

\bibitem{jiao2025addressing}
T.~Jiao, Z.~Xiao, Y.~Xu, C.~Ye, Y.~Huang, Z.~Chen, L.~Cai, J.~Chang, D.~He,
  Y.~Guan \emph{et~al.}, ``Addressing the curse of scenario and task
  generalization in {AI-6G}: A multi-modal paradigm,'' \emph{{IEEE} Trans.
  Wireless Commun.}, Apr. 2025.

\bibitem{jiao2025ai2mmum}
T.~Jiao, Z.~Xiao, Y.~Huang, C.~Ye, Y.~Feng, L.~Cai, J.~Chang, F.~Liu, Y.~Xu,
  D.~He \emph{et~al.}, ``Ai2mmum: Ai-ai oriented multi-modal universal model
  leveraging telecom domain large model,'' \emph{arXiv preprint
  arXiv:2505.10003}, 2025.

\bibitem{chen2025analogical}
Z.~Chen, Z.~Zhang, Z.~Xing, R.~Li, Z.~Yang, R.~Jin, C.~Huang, Y.~Yang, and
  M.~Debbah, ``Analogical learning for cross-scenario generalization: Framework
  and application to intelligent localization,'' \emph{arXiv preprint
  arXiv:2504.08811}, 2025.

\bibitem{liu2023wisr}
S.~Liu, Z.~Chen, M.~Wu, C.~Liu, and L.~Chen, ``Wisr: Wireless domain
  generalization based on style randomization,'' \emph{{IEEE} Trans. Mobile
  Comput.}, vol.~23, no.~5, pp. 4520--4532, May 2024.

\bibitem{10207706}
S.~Liu, Z.~Chen, M.~Wu, H.~Wang, B.~Xing, and L.~Chen, ``Generalizing wireless
  cross-multiple-factor gesture recognition to unseen domains,'' \emph{{IEEE}
  Trans. Mobile Comput.}, vol.~23, no.~5, pp. 5083--5096, May 2024.

\bibitem{10742098}
Y.~Zhang, Q.~Li, H.~Liu, L.~Yang, and J.~Yang, ``Domain generalization for
  cross-receiver radio frequency fingerprint identification,'' \emph{{IEEE}
  Internet Things J.}, vol.~12, no.~5, pp. 5207--5218, Mar. 2025.

\bibitem{seretis2022toward}
A.~Seretis and C.~D. Sarris, ``Toward physics-based generalizable convolutional
  neural network models for indoor propagation,'' \emph{{IEEE} Trans. Antennas
  Propag.}, vol.~70, no.~6, pp. 4112--4126, 2022.

\bibitem{li2025digital}
T.~Li, H.~Lei, H.~Guo, M.~Yin, Y.~Hu, Q.~Zhu, and S.~Rangan, ``Digital
  twin-enhanced wireless indoor navigation: Achieving efficient environment
  sensing with zero-shot reinforcement learning,'' \emph{IEEE Open Journal of
  the Communications Society}, Mar. 2025.

\bibitem{zheng2023cell}
Y.~Zheng, J.~Wang, X.~Li, J.~Li, and S.~Liu, ``Cell-level rsrp estimation with
  the image-to-image wireless propagation model based on measured data,''
  \emph{{IEEE} Trans. Cogn. Commun. Netw.}, vol.~9, no.~6, pp. 1412--1423, Dec.
  2023.

\bibitem{li2024map}
Z.~Li, M.~Chen, G.~Li, X.~Lin, and Y.~Liu, ``Map-driven mmwave link quality
  prediction with spatial-temporal mobility awareness,'' \emph{{IEEE} Trans.
  Mobile Comput.}, Dec. 2024.

\bibitem{dosovitskiy2017carla}
A.~Dosovitskiy, G.~Ros, F.~Codevilla, A.~Lopez, and V.~Koltun, ``Carla: An open
  urban driving simulator,'' in \emph{Conference on robot learning}, 2017, pp.
  1--16.

\bibitem{garg2020functional}
S.~Garg and Y.~Liang, ``Functional regularization for representation learning:
  A unified theoretical perspective,'' \emph{Advances in Neural Information
  Processing Systems}, vol.~33, pp. 17\,187--17\,199, 2020.

\bibitem{kim2024spafl}
M.~Kim, W.~Saad, M.~Debbah, and C.~S. Hong, ``Spafl: Communication-efficient
  federated learning with sparse models and low computational overhead,''
  \emph{Advances in Neural Information Processing Systems}, vol.~37, pp.
  86\,500--86\,527, 2024.

\bibitem{tian2021exploring}
J.~Tian, Y.-C. Hsu, Y.~Shen, H.~Jin, and Z.~Kira, ``Exploring covariate and
  concept shift for out-of-distribution detection,'' in \emph{NeurIPS 2021
  Workshop on Distribution Shifts: Connecting Methods and Applications}, 2021.

\bibitem{jaeckel2017explicit}
S.~Jaeckel, L.~Raschkowski, S.~Wu, L.~Thiele, and W.~Keusgen, ``An explicit
  ground reflection model for mm-wave channels,'' in \emph{Proc. IEEE Wireless
  Commun. and Networking Conf.}, San Francisco, CA, USA, May 2017, pp. 1--5.

\bibitem{3gpp}
``3gpp tr 38.901 v14.2.0, “study on channel model for frequencies from 0.5 to
  100 ghz”,'' 2019.

\bibitem{bsm}
\BIBentryALTinterwordspacing
{U.S. Department of Transportation}. Advanced messaging concept development:
  Basic safety message. [Online]. Available:
  \url{https://data.transportation.gov/Automobiles/Advanced-Messaging-Concept-Development-Basic-Safet/eezi-v4pm/about_data}
\BIBentrySTDinterwordspacing

\bibitem{tezergil2022wireless}
B.~Tezergil and E.~Onur, ``Wireless backhaul in 5g and beyond: Issues,
  challenges and opportunities,'' \emph{IEEE communications surveys and
  tutorials}, vol.~24, no.~4, pp. 2579--2632, 2022.

\bibitem{conlan2017blender}
C.~Conlan, ``The blender python api,'' \emph{Precision 3D Modeling and Add},
  2017.

\bibitem{ali2020passive}
A.~Ali, N.~Gonz{\'a}lez-Prelcic, and A.~Ghosh, ``Passive radar at the roadside
  unit to configure millimeter wave vehicle-to-infrastructure links,''
  \emph{{IEEE} Trans. Veh. Technol.}, vol.~69, no.~12, pp. 14\,903--14\,917,
  Dec. 2020.

\bibitem{xu2022computer}
W.~Xu, F.~Gao, X.~Tao, J.~Zhang, and A.~Alkhateeb, ``Computer vision aided
  mmwave beam alignment in v2x communications,'' \emph{{IEEE} Trans. Wireless
  Commun.}, vol.~22, no.~4, pp. 2699--2714, Apr. 2022.

\bibitem{narayanan2020lumos5g}
A.~Narayanan, E.~Ramadan, R.~Mehta, X.~Hu, Q.~Liu, R.~A. Fezeu, U.~K. Dayalan,
  S.~Verma, P.~Ji, T.~Li \emph{et~al.}, ``Lumos5g: Mapping and predicting
  commercial mmwave 5g throughput,'' in \emph{Proceedings of the ACM internet
  measurement conference}, 2020, pp. 176--193.

\bibitem{yao2022c}
H.~Yao, Y.~Wang, L.~Zhang, J.~Y. Zou, and C.~Finn, ``C-mixup: Improving
  generalization in regression,'' \emph{Advances in neural information
  processing systems}, vol.~35, pp. 3361--3376, 2022.

\bibitem{10171192}
M.~Kim, W.~Saad, M.~Mozaffari, and M.~Debbah, ``Green, quantized federated
  learning over wireless networks: An energy-efficient design,'' \emph{{IEEE}
  Trans. Wireless Commun.}, vol.~23, no.~2, pp. 1386--1402, Feb. 2024.

\end{thebibliography}
\end{document}